\documentclass[preprint,12pt]{elsarticle}

\usepackage{amssymb,makeidx,latexsym,rotate,epsf,amssymb,graphicx,color, psfrag, subfigure, url, setspace, amsthm,amsmath, multirow, booktabs}
\usepackage{natbib}
\newtheorem{example}{Example}
\newtheorem{theorem}{Theorem}
\newtheorem{lemma}{Lemma}
\newtheorem{definition}{Definition}
\newtheorem{remark}{Remark}

\journal{Information Sciences}

\begin{document}

\begin{frontmatter}



\title{Optimal design, financial and risk modelling with stochastic processes having semicontinuous  covariances}


\author[IFAS, VAL]{M. Stehl\'\i k\corref{cor1}}
\ead{mlnstehlik@gmail.com}
\author[IFAS]{Ch. Helpersdorfer}
\author[IFAS]{P. Hermann}
\ead{philipp.hermann@jku.at}

\address[IFAS]{Department of Applied Statistics, Johannes Kepler University Linz, Altenbergerstra\ss e 69, 4040 Linz, Austria}
\address[VAL]{Institute of Statistics, University of Valpara\'iso, Blanco 951, Valpara\'iso, Region de Valpara\'iso, Chile}

\cortext[cor1]{Corresponding author}

\begin{abstract}
A.N.~Kolmogorov proposed several problems on stochastic processes, which has been rarely addressed later on. One of the open problems are stochastic processes with discontinuous covariance function. For example, semicontinuous covariance functions have been used in regression and kriging by many authors in statistics recently. In this paper we introduce purely topologically defined regularity conditions on covariance kernels which are still applicable for increasing and infill domain asymptotics for regression problems, kriging and finance.  These conditions are related to semicontinuous maps of Ornstein Uhlenbeck (OU) processes.  Beside this new regularity conditions relax the continuity of covariance function by consideration of semicontinuous covariance. We provide several novel applications  of the introduced class  for optimal  design of random fields, random walks in finance and probabilities of ruins related to shocks, e.g.~by earthquakes. In particular we construct a random walk model with semicontinuous covariance.
\end{abstract}
\begin{keyword}
Optimal design \sep semicontinuous  covariance \sep correlated process \sep strong convergence \sep topology


\end{keyword}

\end{frontmatter}

\textbf{Dedication:}  \emph{This paper discusses processes with jumps in correlations, issue dedicated to  Andrey N. Kolmogorov.}

\section{Introduction} \label{intro}
Mostly used dependence structure in regression problems is covariance. \cite{Nather85} states in his book: ``One of the fundamental assumptions, the knowledge of the covariance function, is in most cases almost unrealistic. It seems to be artificial, that the first moment $E(Y(x))$ is assumed to be unknown whereas the more complicated second one is assumed to be known...''. Recently, discontinuous covariance functions have been used in both regression and kriging by many authors. For instance, it was observed that in some situations it can be useful to use semicontinuous covariance functions instead of continuous ones. One example is an application of input deformations with Brownian motion filters for discontinuous regression (see \cite{Girdziusas}). For a treatment of discontinuous nature of kriging interpolation see e.g.~\cite{Meyer}.  Another widely spread use of semicontinuous covariances follows from usage of nugget effect, see \cite{Cressie}.
In computer experiment literature, typically used covariance functions vary from analytically smooth Gaussian to one-sided differentiable exponentially decaying (see \cite{SacksW}). However, less smooth covariance have not been studied from statistical perspective.

But how erratic may a covariance function be? If we would like to consider the pragmatical point of view, then any practically relevant discontinuous covariance function should be measurable. Then using the result of \cite{Crum}, measurable covariance function $C$ admits decomposition $C=C_0+C_1,$ where $C_0$ is  a continuous covariance and $C_1$ vanishes Lebesgue almost everywhere. In the latter we assume (without loss of generality) domination by a Lebesgue measure, which covers most of practical applications. However, it is essential to recall the assumption of measurability of $C$ (see \cite{Crum}). \cite{Crum} has also proven that if $C$ is isotropic and positive definite on $R^m,m>1$ then  $C$ is continuous except perhaps at $d=0.$ However, for $m=1$ the latter is not true anymore, as an example we may take $C(q)=1$ for all rational numbers $q\in Q$ and $0$ otherwise, which is positive definite, isotropic and discontinuous on $Q.$ In this  paper we consider the more complicated, but very practical, case $d=1$ (e.g.~time).

Beside the above discussed problem, the path restriction from continuity of covariance is well recognized (at least from \cite{BELAYEV} to the best knowledge of author, even though  Kolmogorov mentioned this problem already). \cite{BELAYEV} has proven that for a stationary Gaussian process with a continuous correlation function, assuming real values, one of the following alternatives holds: either, with probability one, the sample functions are continuous or, with probability one, they are unbounded in every finite interval. But this is too restrictive for many practical statistical problems. As a remedy, we suggest the semicontinuity of covariance functions as a good and practical substitute for a continuity. More precisely, semicontinuity is a more  appropriate variant  of a ``continuity'' framework,  which can still justify increasing domain asymptotics under mild regularities on space. We provide weak conditions on covariance functions which lead to feasible results for increasing domain asymptotics. We represent a class of processes obtained by specific semicontinuous maps of covariance functions of stationary Ornstein-Uhlenbeck (OU) processes. Several examples aim at convincing the reader that this
way of approach to regularity conditions for covariance functions opens a direct way to relax continuity  conditions in order to benefit statistical science.

\subsection{Model}

Let us consider for the sake of simplicity the isotropic stationary process (see e.g.~\cite{Cressie}) $$ Y\left( x\right) = \theta +\varepsilon \left( x\right) $$ with the design points $x_1,...,x_n$ taken from a compact design space $X$. For the sake of simplicity we consider $X=[a,b].$ The mean parameter $E(Y(x)):=\theta\in \Theta $ is unknown, the variance-covariance structure $C_r(d)$ depends on another unknown parameter $r\in \Omega $ and $d_{i}$ is the distance between two particular design points, $x_i$ and $x_{i+1}.$ Parametric spaces $\Theta$ and $\Omega $ are open sets with respect to standard topology.  Let  us assume $E\left(Y(s+h)-Y(s)\right)=0$ and
define $$2\gamma(h)=\mbox{var}\left(Y(s+h)-Y(s)\right)$$
(equation make sense only when right side depends only on $h$). If this is the case, we will say that the process is \emph{intrinsically stationary}, the function $2\gamma(h)$ is then called \emph{variogram} and $\gamma(h)$ is called \emph{semivariogram} (for more details see \cite{Cressie}). Let us briefly introduce the most common isotropic semivariograms, for further discussion see e.g.~\cite{Cressie}. Now consider the three basic isotropic models: linear, spherical and exponential.

\begin{itemize}

\item \textbf{Linear model} valid in $R^n,\ n\geq 1,$
\begin{eqnarray}
\gamma(d)=\left\{\begin{array} {ll}
  \tau^2+\sigma^2d ,&\mbox{for  $d>0$,}\\
      0 &\mbox{otherwise.}\end{array}\right.\label{eq:lin}
       \end{eqnarray}

\item \textbf{Spherical model} valid in $R^1,\ R^2, R^3$
 \begin{eqnarray}
\gamma(d)=\left\{\begin{array} {lll}
  \tau^2+\sigma^2(\frac{3d}{2r}-\frac12 (\frac dr)^3) ,&\mbox{for  $0<d\leq r$,}\\
 \tau^2+\sigma^2,&\mbox{for  $r<d$,}\\
      0 &\mbox{otherwise.}\end{array}\right.\label{eq:sph}
       \end{eqnarray}

\item \textbf{Exponential model} valid in $R^n,\ n\geq 1,$
\begin{eqnarray}
\gamma(d)=\left\{\begin{array} {ll}
  \tau^2+\sigma^2(1-\exp(-r d)),&\mbox{for
     $d>0$,}\\
      0 &\mbox{otherwise.}\end{array}\right.\label{eq:exp}
       \end{eqnarray}
\end{itemize}

The common property of all three isotropic covariance models is that the covariances decrease with distance of design points. This motivates the
introduction of a class of semicontinuous, but non-increasing covariances, which will be called later abc, defined as follows.

\begin{definition}\label{def1} \textbf{(abc class)}

We assume the class of positive definite functions $C_r(d):\Omega\times R^+\to R$ such that

a) $C_r(0)=1,\ C_r(d)\geq 0$ for all $r\in \Omega$ and $0<d<+\infty,$

b) for all $r$  mapping $d\to C_r(d)$ is semicontinuous, almost everywhere convex and
decreasing on $(0,+\infty)$

c) $\lim_{d\to +\infty}C_r(d)=0.$
\end{definition}

\begin{remark}

\textbf{i)} First, notice that conditions abc in Definition \ref{def1} define a class of valid covariance functions for continuous $C_r(d)$ (due to
the celebrated criterion of  \cite{Polya}). However, not every semi-continuous, but almost everywhere convex and non-increasing map $C_r(d)$ defines a valid correlation structure in $(0,+\infty)^n,$ since positive definiteness may be violated.  Therefore we have written ``positive definite" explicitly in Definition \ref{def1}.  However, it is easy to check that  for $n=2$ abc define always a valid correlation,  since $C_r(d)$ is decreasing. For higher $n$ one can use the fact that $C_r(d)$ has {maximally countable many points} of discontinuity, which will be well separated by a sufficiently regular support topology  (which may allow us to make a modified  proof from \cite{Polya}). A more detailed classification of such classes is out of scope of this paper and will be a valuable direction for further research.

\bigskip

\textbf{ii)} Relaxation of condition b) from (typically assumed) continuity to semicontinuity is worth some words. First, from practical point of view, many covariance structures recently used in theory and practice have been indeed discontinuous and still semicontinuous. A survey to semicontinuous functions can be found in  \cite{Neubrunn}. To illustrate semicontinuous (but not continuous) covariance functions (i.e.~covariance function satisfying abc) let us consider for appropriately chosen $c$ and $D$ the covariance defined at $R^+$
\begin{eqnarray} C(d)=\left\{\begin{array} {ll} \sigma^2,&\mbox{for
$d=0$,}\\
c\sigma^2,&\mbox{for
$0<d\leq D$,}\\
0, &\mbox{else}\end{array}\right.\label{disccov}
\end{eqnarray}
which is the prolongation of the covariance given by $c\in [0,1],$ $D,d\in$ $[0,2]$ in \cite{MuellerPazman}.
Figure \ref{fig:Paths2} contains several trajectories for different values of parameters $c$ and $\sigma$ of random walk $Y_t=\sum_{i=1}^t X_i$ based on Gaussian process $X_i$ with covariance (\ref{disccov}).
For this computation we fixed $D=2, d=0.1.$ The impact of the parameters $\sigma$ and $c$ should be visible within one plot. Therefore the comparative line (black) for the path of $Y_t$ was computed for $c = 0.1$ and $\sigma = 1$. It is clear that increasing $c$ yields in a shift in direction of the origin as well as the effect of a decrease of $\sigma$.
\begin{figure}[!ht]
	\centering
		\includegraphics[width=0.45\textwidth, trim = 0 20 0 0, clip]{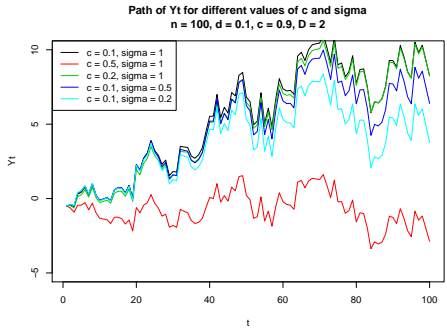}
	\caption{Paths for different values of parameters c and $\sigma$ on the basis of equation (\ref{disccov})}
	\label{fig:Paths2}
\end{figure}

The next theorem provides
the representation of covariance functions which are semicontinuous
maps of covariance of Ornstein Uhlenbeck process.

\begin{theorem} \textbf{(Representation Theorem)}

Let $C$ be abc. Then $C_r(d)=\sigma^2\exp(-\psi_r(d)),$ where
$\psi_r: [0,+\infty)\to R\cup \{+\infty\}$ is: semicontinuous, non
decreasing, $\lim_{d\to +\infty}\psi_r(d)=+\infty.$
\end{theorem}

\begin{proof}{(Representation Theorem)}

We have
\begin{eqnarray}
\psi_r(d)= -\log\left(\frac{C_r(d)}{\sigma^2}\right), \label{komposicia}
\end{eqnarray}
where $\sigma^2=C_r(0).$ Thus semicontinuity  follows from
(\ref{komposicia}), since it is a composition of a continuous and
semicontinuous function. Non-decreasingness follows from
(\ref{komposicia}), since it is a composition of a monotonous and
non-increasing function. We have $lim_{d\to
+\infty}\psi_r(d)=-\lim_{d\to +\infty}\log C_r(d)=-log(0)=+\infty.$
\end{proof}

\bigskip

\textbf{iii)} Notice that  assumptions abc are fulfilled with many covariance structures (e.g.~power exponential correlation family or Mat\'ern class, see e.g.~\cite{SacksW}).

\textbf{iv)}
Let $C_1,C_2\in abc$ then for any $\alpha, \beta>0$ $\alpha C_1+\beta C_2\in abc.$

Let $C\in abc$, then for any $\alpha\in N$ also $ C^\alpha\in abc.$ Here $N$ denotes the set of all natural numbers.

\end{remark}

The paper is organized as follows: in section 2 we
study properties of abc class of covariance functions in more details.  Also a discussion on the topological convergences preserving the regularity conditions abc is given for different choices of the target space regularities. Pointwise convergent sequences of abc-covariance  in metric space is proven to be convergent to the abc-covariance. Strong convergence, a purely topologically defined uniform convergence for non-uniform spaces is found to be an appropriate candidate for preserving conditions abc when target space is endowed with general topology. Section 3 considers optimal design for regression problem with correlated errors for abc covariance function. In particular,  the lower and upper bound for $M_\theta(d)$ for the class of processes with covariances satisfying abc is given and the monotonicity of functions $d\to LB(d),M_\theta(d), UB(d)$ is proven. Optimality of equidistant design is proven for parameter $\theta$ and covariance functions from abc class. We also provide conditions for existence of admissible optimal design for parameter $r$ in abc class. Several examples illustrate increasing domain asymptotics for various covariance functions. In  section 4 we provide illustrations of introduced stochastic processes and their applications  to finance. The first application is  forecasting of stock markets with implications to Efficient Market Hypothesis. The second application is to  model the probability of ruin with semicontinuous covariance, with potential importance regarding catastrophes such as earthquakes or rain-storms.
To maintain the continuity of explanation, the technicalities and proofs are put into the Appendix.

\section{abc class}

\subsection{Topological convergences preserving class abc}

Here we provide results on topological convergences preserving regularity conditions abc. We prove that the pointwise convergence is sufficient for the locally separable metric domain and metric target space (usually both are $R$ with Euclidean topology). However, there are situations where it is more appropriate to endow the target space $R$ by a different topology: for this case we give regularity conditions on the topology under which abc is preserved by strong convergence. We show that the minimal regularity of the target space is to be a regular and fully normal topological space.

We have found strong convergence (as introduced by \cite{KupkaToma}) to be the most relevant for our setup. It is defined purely topologically and it is a topological uniform convergence for non-uniform spaces as the analogue of uniform convergence for metric and uniform spaces which preserve the continuity of the functions. Strong convergence can be defined by the convergent nets and has many nice properties, e.g.~preserving the fixed point property (FPP) (see \cite{Kupka00}).  The strong convergence is the appropriate one, since the uniform convergence in a non compact space does not
preserve FPP (see \cite{Kupka00}). In such a context strong convergence is the weakest one with FPP: pointwise convergence does not have such a property and convergences based on graph topology, fine and open cover are ``too strong'', i.e.~FPP is also preserved by the graph topology, but the graph convergence implies the strong convergence. Theorem \ref{th2} shows that properties abc are preserved by pointwise convergence for the metric target space.

\begin{theorem} \label{th2}
Let $X_n$ be a sequence of isotropic random fields with covariance kernels $K_n$ defined on a locally separable metric space mapping to the metric space satisfying conditions abc. Let $K_n$ be uniformly  convergent to $K,$ then $K$ also satisfies abc.
\end{theorem}

\begin{remark}
Notice, that if we have a continuous covariance, we will need to have  a uniform convergence to preserve continuity.
\end{remark}

Now let us consider the general situation, i.e.~general topology is endowed on the target space. Theorem \ref{theo1} gives the preservation of abc for continuous covariances.

\begin{theorem}\label{theo1}
Let $\{X_\gamma\},\gamma\in \Gamma $ be a net of isotropic random fields with continuous covariance kernels $K_\gamma$ mapping to the regular topological space satisfying conditions abc. Let $K_\gamma$ be strongly convergent to $K,$ then $K$ also satisfies abc.
\end{theorem}

Before proving the general theorem \ref{theo1} for the case of abc covariance functions we need following two technical Lemmas. Their proofs  can be found in \cite{Stehlik2014}.
\begin{lemma}
Let $X$ be a topological space and $Y$ be a regular space. Let $\{f_\gamma:X\to Y,\gamma\in \Gamma\}$  be a net of functions semicontinuous at $x\in X$ that converges strongly to a function $f$. Then $f$ is semicontinuous at $x$.
\end{lemma}

Example 2 in \cite{KupkaToma} shows, that the regularity of target space $Y$ cannot be omitted. Now we need a Lemma providing a more general version of interchange of limits theorem for a net of not necessarily continuous covariance functions satisfying abc to justify that limit also satisfies c).

\begin{lemma}
Let $Z$ be a topological space, $Y$ is a fully normal, $T_1$-space and $\emptyset\ne X\subseteq Z.$ Let $a\in Z$ be an accumulation point of a set $X.$ Let net $f_\gamma (:X\to Y) \stackrel{s}{\to}f$ and $\forall \gamma\in \Gamma$ exists $\lim_{x\to a}f_\gamma(x):=A_\gamma\in Y.$ Then a net $\{A_\gamma\}$ is convergent and the limits interchange is valid, i.e.
$$ \lim_{x\to a}\lim_{\gamma\in \Gamma} f_\gamma(x)= \lim_{\gamma\in \Gamma} \lim_{x\to a}f_\gamma(x)$$
\end{lemma}

Now we are ready to formulate the general theorem for preserving abc class.

\begin{theorem}
Let $\{X_\gamma\},\gamma\in \Gamma $ be a net of isotropic random fields with covariance kernels $K_\gamma$ mapping to the regular and fully normal topological space satisfying conditions abc. Let $K_\gamma$ be strongly convergent to $K.$ Then $K$ also satisfies abc.
\end{theorem}

\section{Optimal Design for regression with correlated errors from abc class}

The determination of optimal designs for models with correlated errors is substantially more difficult and for this reason not so
well developed. Here we concentrate on Ornstein-Uhlenbeck processes. For the influential papers there is a pioneering
work of \cite{Hoel}, who considered the weighted least square estimate, but considered mainly equidistant designs.
Optimal design for estimation of parameters of OU processes has been studied in \cite{KS}, optimal designs for prediction of OU sheets in \cite{BaranSS}.
We can find applications of various criteria of design optimality for second-order spatial models in the literature. Since in our setup the information matrix is scalar, maximizing of $M_\theta$ leads to the optimal design in the sense of D, E, A or G optimality. Theoretical justifications for using the Fisher information for $D$-optimal designing under correlation can be found in \cite{Abt, Pazman07,Zhu}.
In our setup, Fisher information matrix for trend parameter $\theta$ is defined as
\begin{equation}\label{Mtheta}
M_\theta(n)=1^TC^{-1}\left( r \right) 1,
\end{equation}
where $n$ denotes the number of design points and ${\cal D}$ is the vector of distances.
According to the results of \cite{Pazman07}  the Fisher information matrix on $r$ has the form
\begin{equation}\label{Mr}
M_{r}(n):=\frac 12 tr \left\{C^{-1}(n,r)\frac{\partial C(n,r)}{\partial
    r}C^{-1}(n,r)\frac{\partial C(n,r)}{\partial r} \right\}.
\end{equation}

Notice, that Ornstein-Uhlenbeck process is a special case of class abc for $\psi_r(d)={rd}.$ For  Ornstein Uhlenbeck processes all distance gradients of $M_\theta$ increase with same speed, since FIM has the form $M_\theta(n,{\cal D})=1+\sum_{i=1}^{n-1}\frac{\exp(r d_i)-1}{\exp(r d_i)+1}$ (see \cite{KS} or \cite{Zagoraiou}). Thus the optimal design is equidistant at any fixed compact design space.

The FIM for both parameters is $M_rM_\theta$ and collapses for the OU process for two point designs, nevertheless it gives a  design with a finite distance for $n>2.$ This is proven in the following:

\begin{lemma}
The distance of neighboring points  in the optimal design for estimation of parameters $(r,\theta)$ is collapsing for $n=2$ and it is equidistant for $n>2.$
\end{lemma}

\subsection{Example:  nugget effect and  the efficiency}

This section considers how the nominal level of Fisher information of the design without nugget would
be affected by a nugget.
We consider class of abc covariances of the form:
\begin{equation}\label{cov1Jump}
Cov(x_s,x_t)=\left\{\begin{matrix}1&\ldots s=t\\
{c}\;exp^{-r\left (|  t-s\right |)} &\ldots s\neq t
\end{matrix}, \right.
\end{equation}
where $0<c\leq 1$ regulates nugget. For $c=1$ we receive a standard OU covariance without nugget, however, $c<1$
introduce nugget $\tau^2=1-c.$
Let $M_{\theta,c}, M_{r,c}$ denote the Fisher information for trend $\theta$ and covariance parameter $r$, respectively. We define  \emph{effectiveness}   in the following form
\begin{equation}\label{effectiveness}
M_{\theta,c}/M_{\theta,1}, M_{r,c}/M_{r,1}.
\end{equation}
We call (\ref{effectiveness}) effectiveness, since it is similar to efficiency, where
the ratio of Fisher information for design and optimal design is computed.
Figure \ref{Effectiveness} illustrates the behavior of \emph{effectiveness} (\ref{effectiveness})
at the design space  $X = [0,1]$. Here we denote distance between design points $d=y-x$ and fix the parameters $r = 1$ and $\theta = 1$.
It is clear that decrease of covariance have decreasing effect on \emph{effectiveness}.
\begin{figure}[!ht]
				\subfigure[Effectiveness for $\theta$ of linear model (\ref{eq:lin}), comparison with one jump]{\includegraphics[width = 0.4\textwidth]{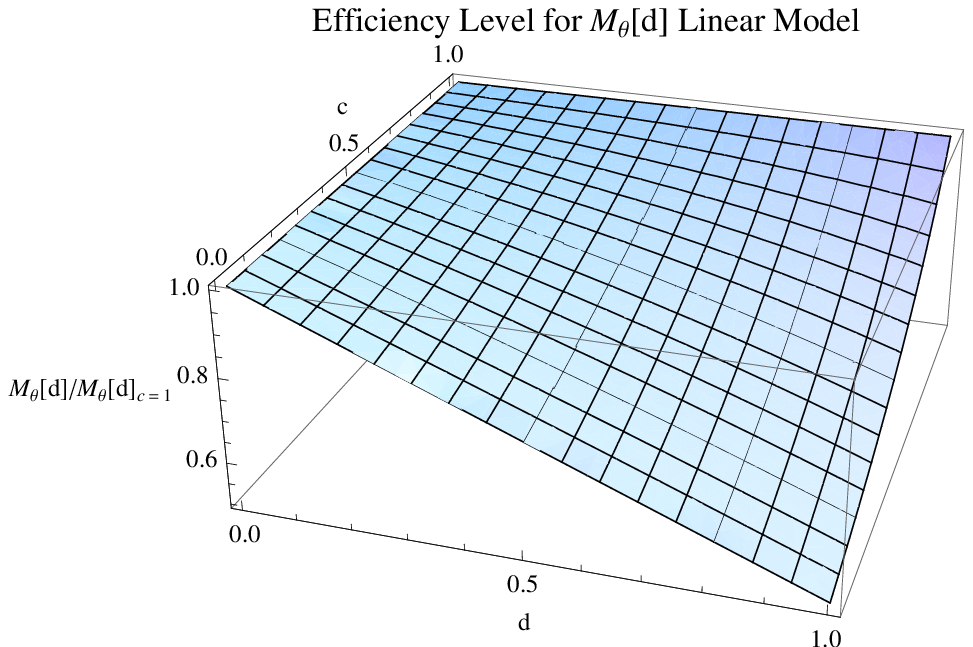}}\hfill
		\subfigure[Effectiveness for $r$ of linear model (\ref{eq:lin}), comparison with one jump]{\includegraphics[width = 0.4\textwidth]{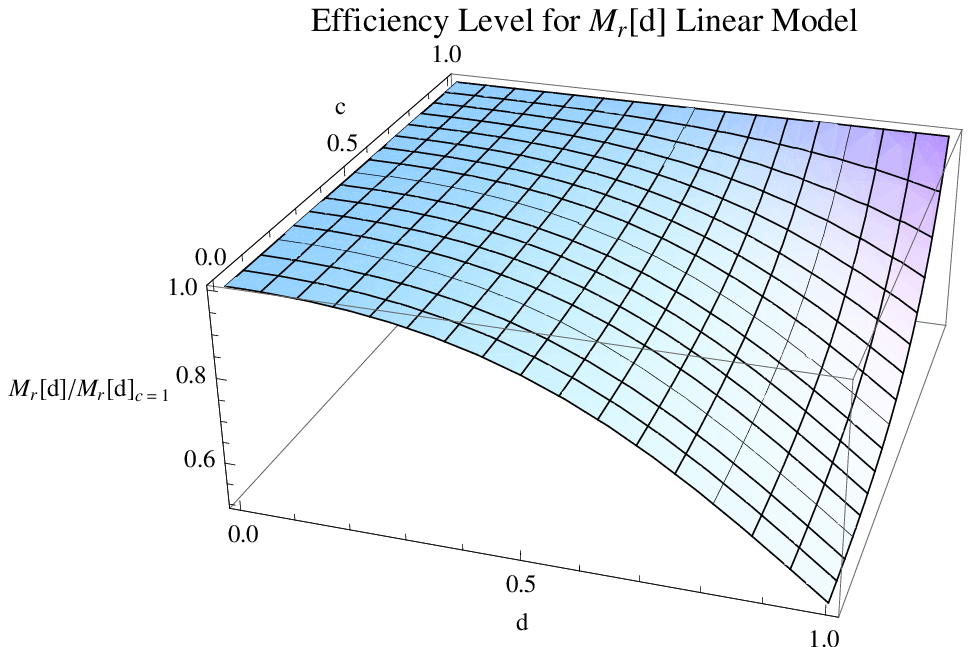}}
				\subfigure[Effectiveness for $\theta$ of spherical model (\ref{eq:sph}), comparison with one jump]{\includegraphics[width = 0.4\textwidth]{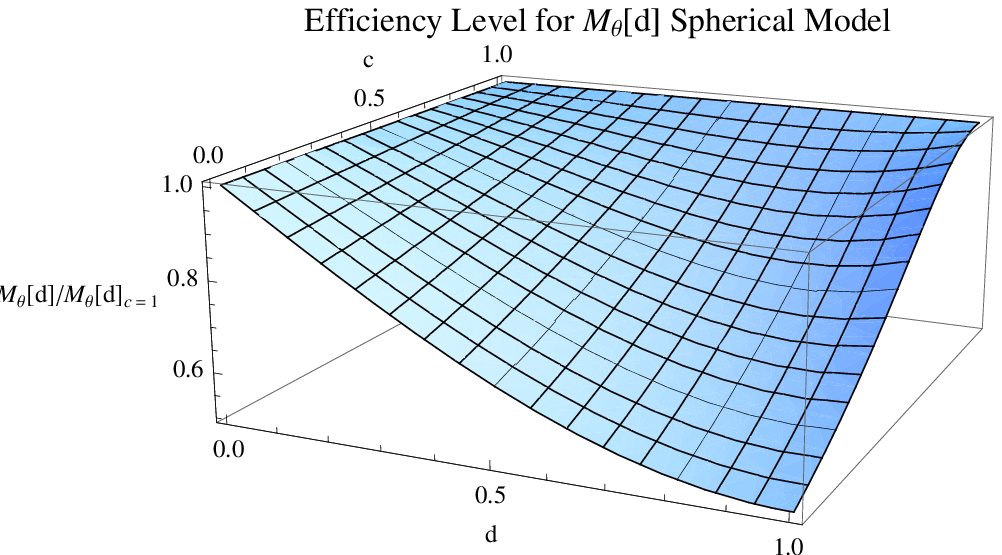}}\hfill
		\subfigure[Effectiveness for $r$ of spherical model (\ref{eq:sph}), comparison with one jump]{\includegraphics[width = 0.4\textwidth]{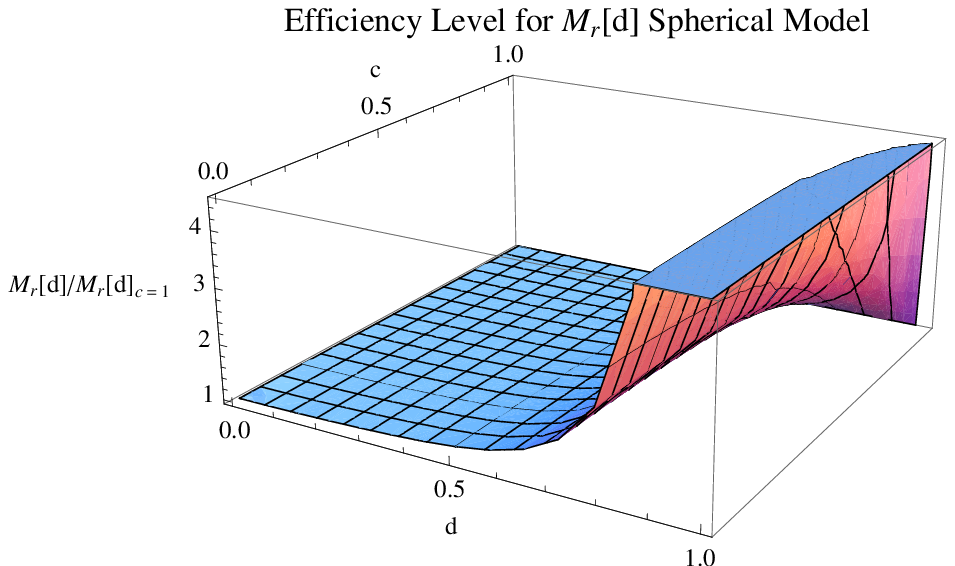}}
				\subfigure[Effectiveness for $\theta$ of exponential model (\ref{eq:exp}), comparison with one jump]{\includegraphics[width = 0.4\textwidth]{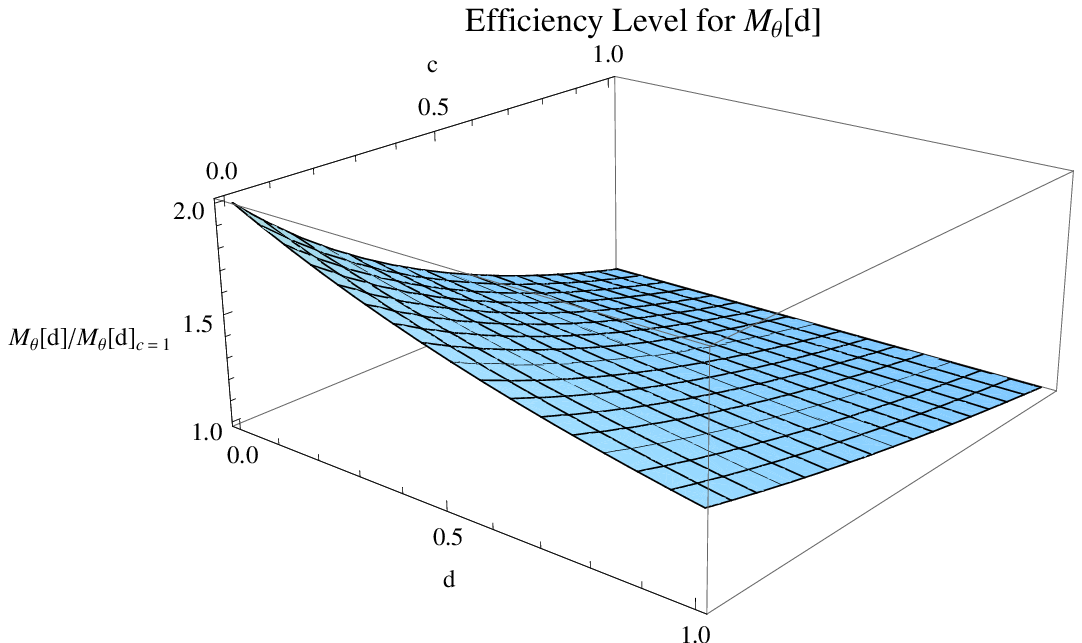}}\hfill
		\subfigure[Effectiveness for $r$ of exponential model (\ref{eq:exp}), comparison with one jump]{\includegraphics[width = 0.4\textwidth]{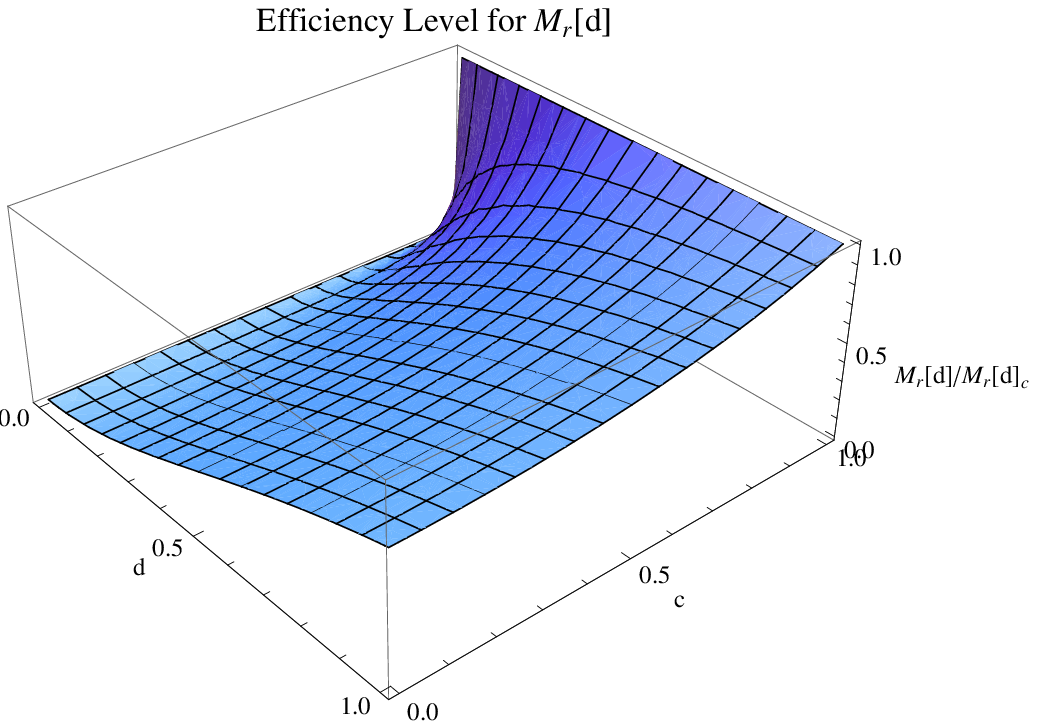}}

	\caption{Effectiveness for linear model shown in (a), (b); the spherical model in (c), (d) and the exponential model in (e), (f). }
	\label{Effectiveness}
\end{figure}

\subsection{Behavior of $M_\theta(n)$}

Hereafter, we introduce the lower and upper bounds for $M_\theta$ and study their properties, together with properties of $M_\theta.$ Let us consider a lower and upper bound for $M_\theta(n, {\cal D})$ of the forms
$$LB(n, {\cal D}):=n\inf_{x}\frac{x^TC^{-1}x}{x^Tx}.$$
$$UB(n, {\cal D}):=n\sup_{x}\frac{x^TC^{-1}x}{x^Tx}.$$
It is easy to see that $LB(n, {\cal D})\leq M_\theta(n, {\cal D})\leq UB(n, {\cal D}).$ The following theorem holds.

\begin{theorem}\label{theorem1}
 \textbf{i)} Let $C_r(d)$ be a covariance structure satisfying abc. Then for any design $\{x,x+d_1,x+d_1+d_2,...,x+d_1+...+d_{n-1}\}$ and for any subset of distances $d_{j},j=1,...,n-1.$

a)  the lower bound function $(d_{i_1},...,d_{i_m})\to LB(n, {\cal D})$ is nondecreasing in ${\cal D}.$ The upper bound function $(d_{i_1},...,d_{i_m})\to UB(n, {\cal D})$ is nondecreasing in ${\cal D}.$  Fisher information $(d_{i_1},...,d_{i_m})\to M_\theta(n, {\cal D})$ is nondecreasing in ${\cal D}.$

b) Especially, for any equidistant design ($\forall i:d_i=d$) functions $d\to LB(n, {\cal D}),$ $d\to UB(n, {\cal D})$ and $d\to M_\theta(n, {\cal D})$ are nondecreasing.

\textbf{ii)} Denote by $a(n,n-1)$ the ratio $M_\theta(n, {\cal D})/M_\theta(n-1, {\cal D}).$ Then $\lim_{\forall i:d_i\to +\infty} a(n,n-1)=\frac{n}{n-1}.$

\textbf{iii)} equidistant design is optimal for $\theta$ in abc on every compact design space $X,$ more precisely for $X=[0,1]$ any point $(d_1,...,d_{n-1})$ of a set $\otimes_{i=1}^{n-1}\psi_r^{-1}(L/(n-1))$ such that $d_i\geq 0, \sum d_i\leq 1$ is a set of optimal inter-distances

\textbf{iv)} for a stationary OU there does not exist an admissible design for parameter $r$ (i.e.~optimal design for $r$ is collapsing). However, in abc class {not necessarily} e.g.~nugget effect can bring a regularization and thus admissible designs may exist.

\end{theorem}

\begin{remark} Notice that i) of Theorem {\ref{theorem1}}  shows that the interval over which observations are to be made should be extended as far as possible. This is supporting the idea of increasing domain asymptotics. Also notice, that both results i, ii) of Theorem \ref{theorem1} generalize the findings of \cite{Hoel} and \cite{KS}.

Particularly, to illustrate result i,b) let us consider the equidistant n-point design for parameter $\theta$ of Ornstein Uhlenbeck process. The covariance matrix is Toeplitz with entries $c^{|i-j|},c=\exp(-rd).$ Then we know that $C^{-1}$ is tridiagonal and that $(1-c^2)C^{-1}$ has the entry $-c$ in every upper and sub-diagonal position and has main diagonal entries $1,1+c^2,...,1+c^2,1$ (see \cite{Horn}, Example 13, page 409). Thus
we get $M_\theta(d)=\frac{2-n+ne^{rd}}{1+e^{rd}}$ (see also Lemma 1 in \cite{KS}) and so $M_\theta(n,{\cal D})$ is an increasing function of distance $d.$

 The formula for Fisher information $M_r(n)$ on correlation parameter $r$ has been recently derived (see \cite{Zagoraiou} and \cite{MullerStehlik3}):
$M_r(n)=\sum_{i=1}^{n-1}\frac{d_i^2(e^{2rd_i}+1)}{(e^{2rd_i}-1)^2}.$ There exists no admissible design for $r.$

\end{remark}

\begin{example}
\textbf{Stationary Ornstein Uhlenbeck process with the nugget}

Let us have
$$Y(x_i)=f(x_i,\vartheta)+e(x_i),\ i=1,2,\ x_1,x_2\in X, C_r(d)=e^{-rd},d:=|x_1-x_2|$$
and only covariance parameter $r$ is the  parameter of interest. Then we know that the maximal Fisher information is obtained for $d=0$ (Collapsing effect, \cite{KS}).
To avoid such 'inconvenient' behavior we decrease the non-diagonal elements by multiplying with factor $\alpha,$ $0<\alpha <1.$ By this we include the nugget effect (micro-scale variation effect) of the form
\begin{eqnarray}
\gamma(d,r)= \left\{\begin{array} {ll} 0&\mbox {for
$d=0$,}\\1-\alpha+\alpha
(1-\exp(-rd))&\mbox{otherwise.}\end{array}\right.
\end{eqnarray}
If  $\gamma(d)\to 1-\alpha>0,$ as $d\to 0,$ then $1-\alpha$ has been called the \emph{nugget effect} by \cite{Matheron}. This is because it is believed that microscale variation is causing a discontinuity at the origin.
Then we obtain
$$M_{r}=\frac{\alpha^2d^2\exp(-2dr)(\alpha^2\exp(-2dr)+1)}{(1-\alpha^2\exp(-2dr))^2}.$$

\cite{Test} have proven that the distance $d$ of the optimal design is an increasing function of nugget effect $1-\alpha.$
The nugget effect makes $C_r$ discontinuous, $C_r(0)=1-\alpha$ and $C_r(d)=\alpha\exp(-rd)$ for $d>0.$ Thus $C_r$ is the member of abc
class for $\alpha<1/2.$ We see that in coherence with Theorem 5 iv) the optimal design is not collapsing.
This issue intrinsically relates to "twin-points" design (see \cite{Crary02} and \cite{Crary15}).

\end{example}

\begin{remark} \emph{Loewner optimality:} The most gratifying criterion refers to Loewner comparison of information matrices. It goes hand in hand with estimation problems, testing hypothesis and general parametric model building (see Sections 3.4-3.10 in \cite{Pukelsheim}). Notice that Theorem 5 says that under the regularity conditions abc the Loewner comparison of two information matrices amounts to comparing their distance vectors.
\end{remark}

\begin{example} Let us consider the two point optimal design (i.e.~design maximizing $M_\theta$). Then for the class of decreasing covariances (corresponding to the increasing variograms) we obtain the optimal design with the maximal inter-point distance. More formally, let $\{x,z\}$ be the two point design in compact design space $X\subseteq R^k$ and let us assume increasing semivariogram $\gamma.$
 Then $M_\theta=\frac2{2-\gamma(d)}$ and the design $\{x,z, ||x-z||=diam X\}$ is optimal. The information gained by the optimal design has  the form
$\frac2{2-\gamma(diam X)}.$ Here we remind that many semivariograms are increasing, e.g.~linear, spherical, exponential, Gaussian, rational quadratic among others. There exist also non monotonous semivariograms, e.g.~wave variogram.
\end{example}

\begin{example}To illustrate Theorem 5 \textbf{iii}, let us consider design space $X=[0,1],$ discontinuous covariance function $C(d)=\exp(-rd)$ for $d<1/2$ and $0$ otherwise and the two point design for the sake of simplicity. Then $M_\theta$ is increasing for $d<1/2$ and stands constant otherwise. Therefore both designs $\{0,1/2\}$ and $\{0,1\}$ are optimal (not only equidistant design is optimal).
\end{example}

\begin{example}  Let us digress to the simple linear regression $ Y(x)=\vartheta_1+\vartheta_2x+\epsilon(x),$
 with modified N\"{a}ther covariance structure, studied in \cite{Stehlik04}. There we consider a modification of Example 6.4 discussed in \cite{Nather85} having the design space $X=[-1,1]$ and a covariance function
\begin{eqnarray}
\mbox{C}(d)= \left\{\begin{array} {ll}
\sigma^2(1-\frac{d}{r}) &\mbox{for $d<r,$}\\
0&\mbox{otherwise.}\end{array}\right.\label{eq:no4}
\end{eqnarray}
Notice that for $r<2$ covariance function  $C(d)$ is not differentiable with respect to parameter $r,$ however, both one sided derivatives exist. The process with such a correlation can be thought of as a model for function required to have one sided first order derivatives (see \cite{SacksY66}). The classic Fisher information assumes the differentiability with respect to the parameter. Still, the Fisher information can be well defined over some open set.

Notice, that for $r>2,$ the modified N\"{a}ther covariance structure (\ref{eq:no4}) constitutes on $[-1,1]$ a linear semivariogram structure with $\gamma(d)=\frac{d}{r}.$ Let us assume now, that only the intercept is the parameter of interest and covariance  parameter $r>2$ is fixed. Then design $\{-1,1\}$ is uniformly optimal because only correlated observations are possible. We have $M_\theta=\frac{2r}{2r-(x_n-x_1)}$ for design $\{x_1,x_2,...,x_n\},\ -1\leq x_1<x_2...<x_{n-1}<x_n \leq 1$ and information gained by the optimal design has the form $\max M_\theta=r/(r-1).$ So $\max M_\theta$ decreases with the positive correlation. As far as $C(d)$ decreases with the distance, the optimal distance is maximal. The same concept occurs also in the case of $3$ observations, when both slope and intercept are estimated (see \cite{Stehlik04}). \cite{Nather85} has shown that since the covariance function can be represented as a linear function of responses, a uniformly optimal design is available for estimating $(\vartheta_1,\vartheta_2),$ which concentrates on the points $\{-1,0,1\}.$
\end{example}

It is clear that optimal designs for both parameters $\theta$ and $r$ in any given compact interval are in some sense trade-offs between collapsing and equidistant design, since the optimal strategy for the estimation of the trend parameter $\theta$ conflicts with the one for estimating the correlation parameter $r.$ This may led to compromises like geometric progressive designs (GPD, used e.g.~by \cite{Zagoraiou} for the case of OU process) or compound designs (see e.g.~\cite{MullerStehlik3}).

\section{Illustrations and Applications to Finance}
If one deals with stochastic models in finance, one will naturally encounter the Efficient Market Hypothesis (EMH) which forms the basis for many financial market models of modern portfolio management (\cite{Fama}). The core statement of the EMH is: nobody who is using the available information can achieve permanently above-average returns when dealing on financial markets. At the same time it is supposed that the actors of financial markets act absolutely rationally. The EMH distinguishes between three versions as the weak-form, semi-strong-form and strong form EMH, respectively. For more details see \cite{Fama1970} among others.
%

In all three forms of EMH, actors on financial markets have the same information available, so one cannot achieve  above-average returns permanently (see \cite{Fama1970}). If the financial markets act efficiently as the EMH claims, the financial crisis, which began in 2007 and lasts till today, cannot be explained easily. Can the forming and bursting of financial bubbles be explained by the use of insider information or by irrational decisions of the actors, or should even the EMH be questioned?

\cite{Lim2010} carried out an investigation to the weak-form EMH of stock markets. Aggregated stock exchange indices of 50 states formed the database. From 1.1.1995 until 31.12.2005, the pitches were collected on a daily basis. The examined 11 years led therefore to 2,870 values per country. Afterwards these 2,870 values were divided into rolling windows. These rolling windows represent a part of the time row, including 200 sample points. They are continuously shifted by one day. The 2,670 samples per country were checked by a hypothesis test if they correspond to a random-walk. To evaluate the efficiency of a stock exchange index, the proportion of samples which did not correspond to a random-walk by a confidence level of 5 \%, was taken. The lower the proportion of a country is, the more its stock exchange index corresponds to an efficient market. From the results of \cite{Lim2010} one can suppose that efficiency as stated in the EMH is not given on stock markets - at least on a temporarily prospect.  Here we develop an alternative random walk model with semicontinuous covariance. Such model is based on an Ornstein-Uhlenbeck process whose covariance structure will be modified.

\subsection{Random walk with semicontinuous covariance}\label{ch:model}
This section presents a model based on the abc covariance structure for high frequency data, which is simple in its implementation. The influences of the model parameters on the time series will be analyzed by simulations with the software package R (\cite{R}).
We  adjust the Ornstein-Uhlenbeck process as follows. For $t \in \mathbb{R}$
		\begin{equation}
			Y(t)=\int_{0}^{t}x(s) ds,\; \text{for}\; t\geq0,
		\end{equation}
		\[
			\textbf{x}=(x(0),\ldots, x(t))\sim N(0,\Sigma),\; \text{with}
		\]
		\begin{equation}
			Cov(x(s),x(t))= \text{e}^ {- |  t-s |}
		\end{equation}

The vector $\textbf{x}$ represents realizations of a stationary Ornstein-Uhlenbeck process and so $Y(t)$ is a so called Integrated-Ornstein-Uhlenbeck process. For the following simulations, the discrete parameter space is in the interval $[1,100]$. The distance between two time points $\Delta t$ is defined as 1. Hence, the process changes to its empirical version:
\begin{equation}
			Y_t=\sum_{i=1}^{t}x_{i},\; \text{for}\; t\in T,
		\end{equation}
		\[
			\textbf{x}=(x_1, x_2,\ldots, x_{t})\sim N(0,\Sigma),\; \text{with}
		\]
		\begin{equation}
			Cov(x_s,x_t)= \text{e}^ {- |  t-s |}
		\end{equation}
		
Figure \ref{fig:model1} shows the theoretical autocorrelation function of  $x_i$ and a path of the stochastic process $Y_t$.

\begin{figure}[!ht]
\begin{center}
\subfigure[Theoretical Autocorrelation Function]{\includegraphics[width = 0.45\textwidth, trim = 0 0 0 40, clip]{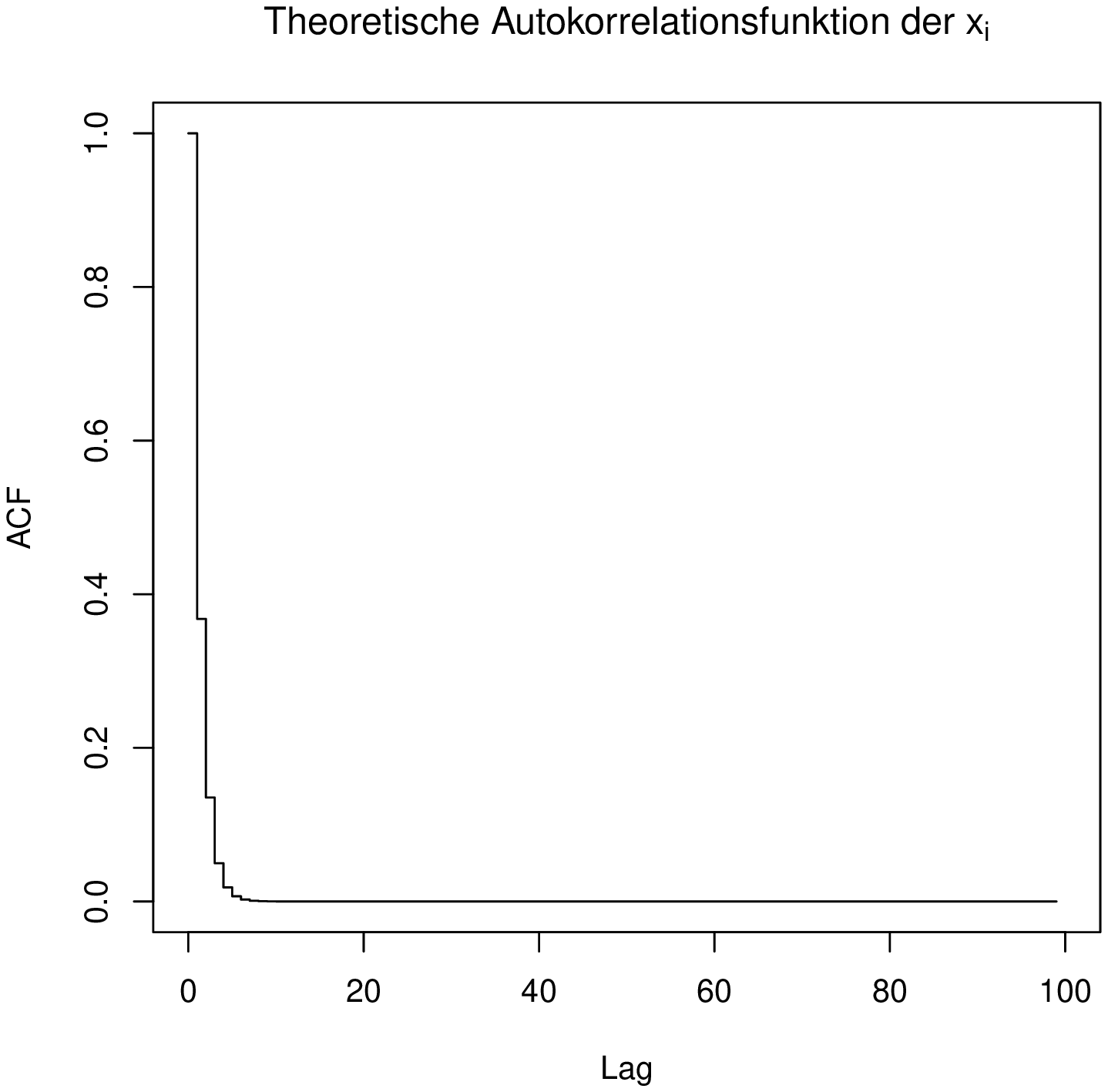}}
\subfigure[Path of $Y_t$]{\includegraphics[width = 0.45\textwidth, trim = 0 0 0 40,clip]{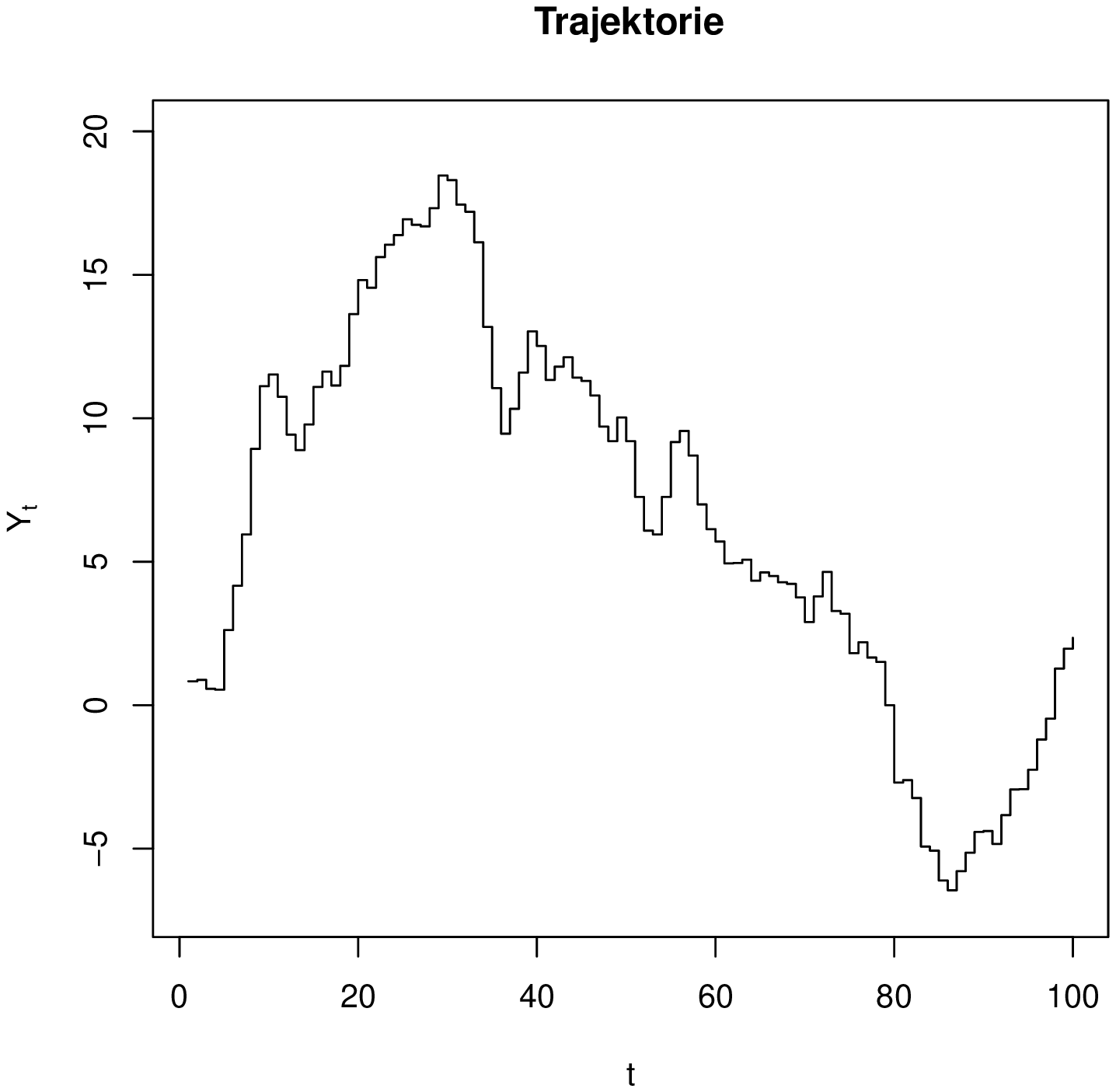}}
\caption{(a) Theoretical autocorrelation function of the $x_i$ and (b) a path of $Y_t$.}
\label{fig:model1}
\end{center}
\end{figure}

The introduced model can be adopted only with limitations for a specific data. Therefore the covariance structure is extended to abc with two parameters $c$ and $r$. The vector $\textbf{x}$ changes to:
\[
\textbf{x}=(x_1, x_2,\ldots, x_{t})\sim N(0,\Sigma),\; \text{with}
\]
\begin{equation}\label{cov1Jump}
Cov(x_s,x_t)=\left\{\begin{matrix}1&\ldots s=t\\
c\;exp^{-r\left (|  t-s\right |)} &\ldots s\neq t
\end{matrix}\right.
\end{equation}

With $r>0$ and $c\in(0,1]$. For $c<1$ the covariance matrix becomes discontinuous. In the following, the influences of these parameters are illustrated by varying them. Figure \ref{fig:model2} (a) and (b) show the influences of $r$, whereas the plots (c) and (d) those of $c$ with the aid of the theoretical autocorrelation function of $x_i$ and a path of $Y_t$.

\begin{figure}[!ht]
\begin{center}
\subfigure[Theoretical autocorrelation function (fixed $c$, varying $r$)]{\includegraphics[width = 0.4\textwidth, trim = 0 0 0 35,clip]{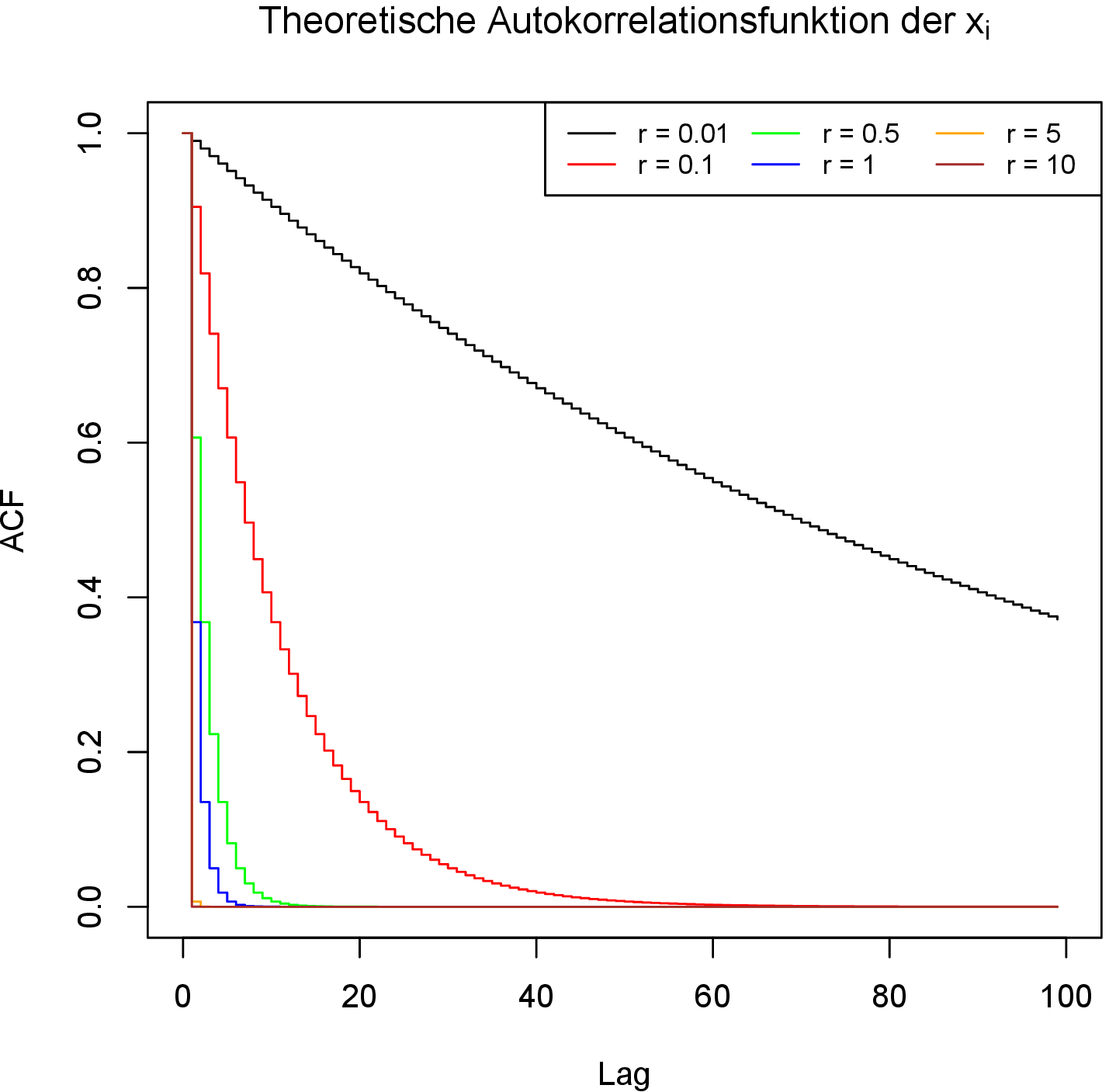}}
\subfigure[Path of $Y_t$ (fixed $c$, varying $r$)]{\includegraphics[width = 0.4\textwidth, trim = 0 0 0 35,clip]{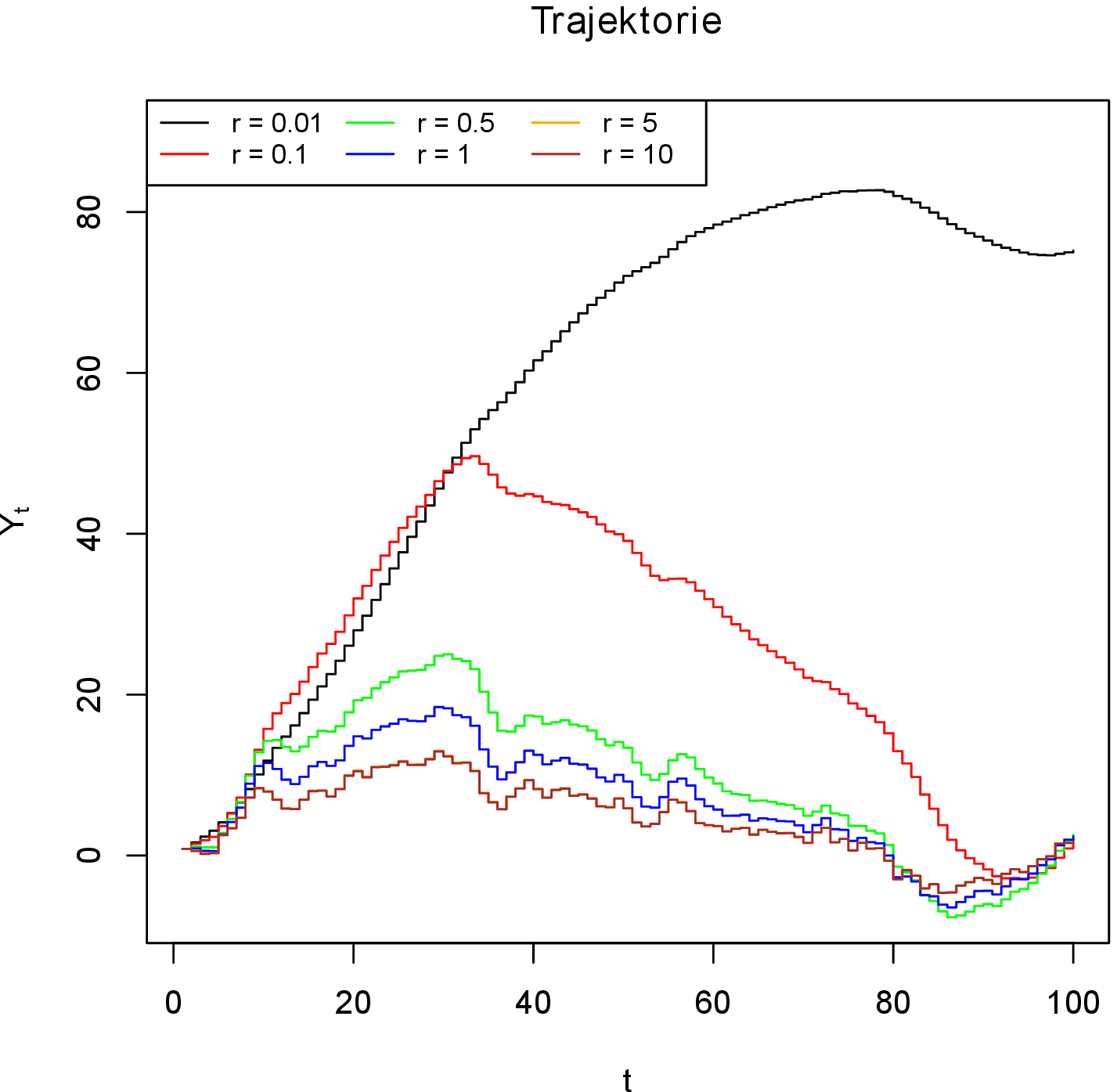}}
\subfigure[Theoretical autocorrelation function (fixed $r$, varying $c$)]{\includegraphics[width = 0.4\textwidth, trim = 0 0 0 35,clip]{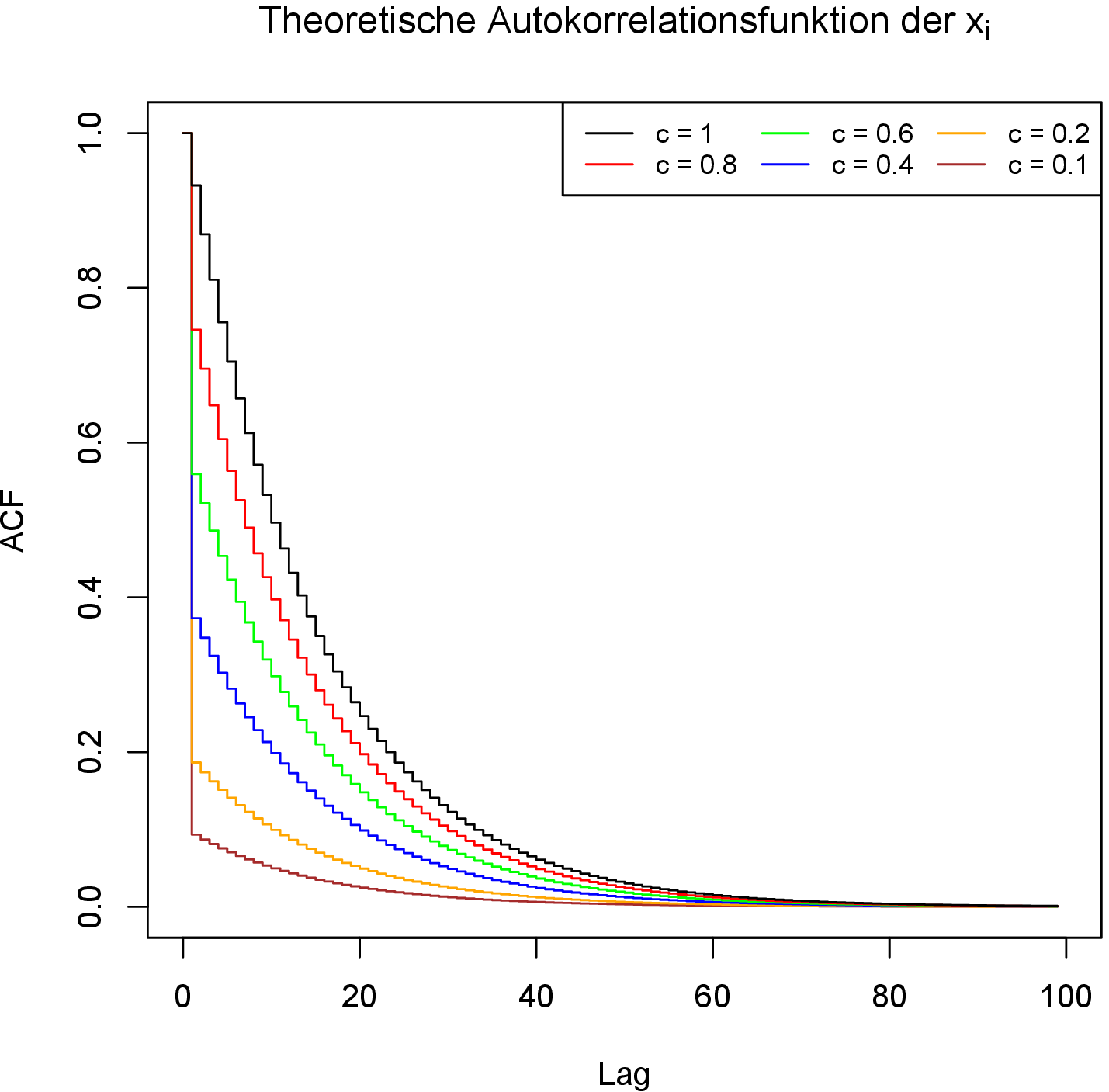}}
\subfigure[Path of $Y_t$ (fixed $r$, varying $c$)]{\includegraphics[width = 0.4\textwidth, trim = 0 0 0 35,clip]{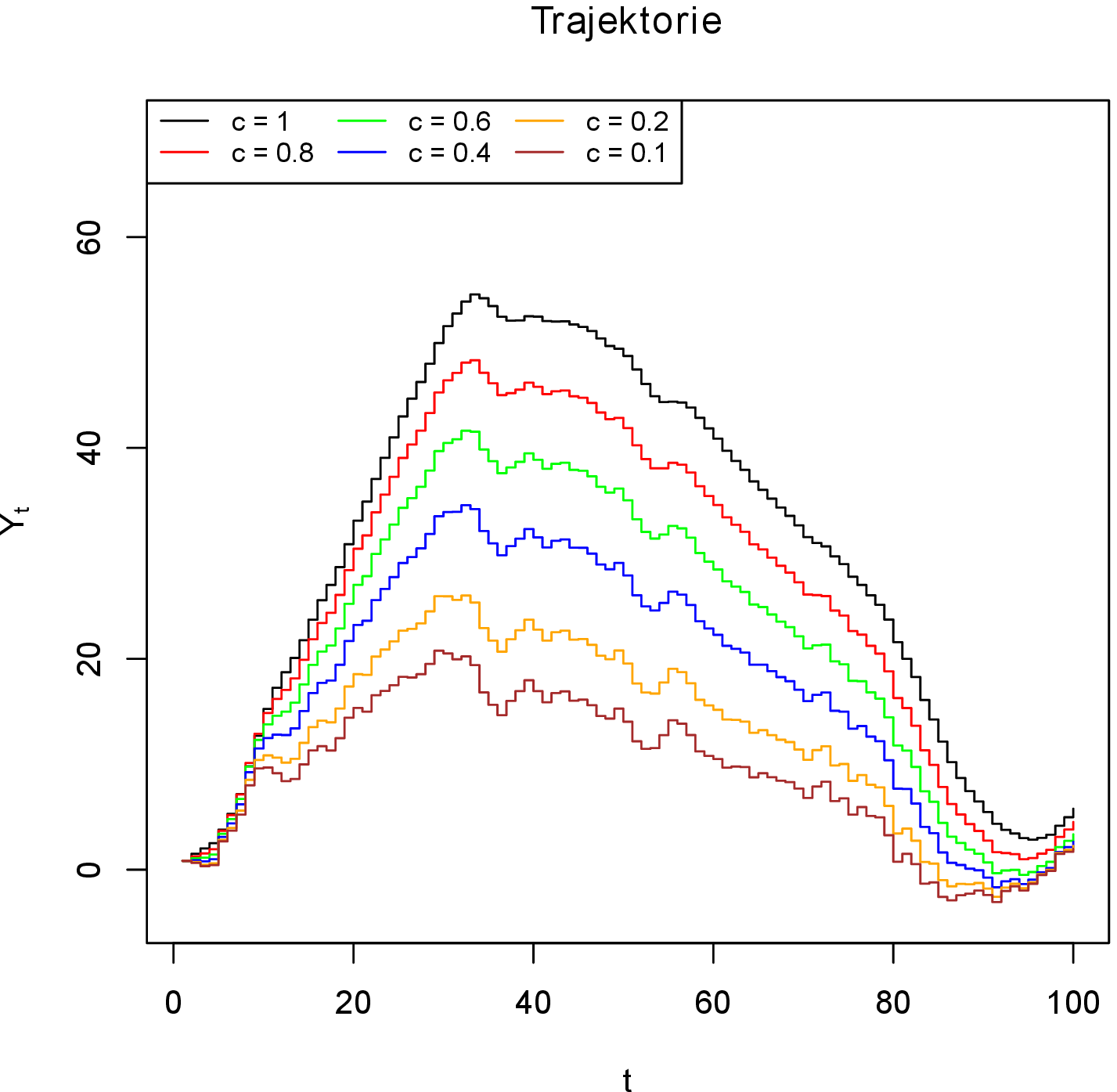}}
\caption{(a) Theoretical autocorrelation function of the $x_i$ and a path of $Y_t$ by varying $c$ with $r=0.07$ (a) and (b). The case of fixed $c=1$ and varying $r$ is plotted in (c) and (d).}
\label{fig:model2}
\end{center}
\end{figure}

Through the parameter $c$ a  jump at Lag 1 of the theoretical autocorrelation function is given. To improve the adjustment of the model on the data, further jumps at different Lags are enabled. The problem at this point is, that the covariance matrix has to be positive semi-definite. The model is expanded, so that up to four jumps in the theoretical autocorrelation function can occur, in order gain information about the influences of the jumps on the process $Y_t$:
		\[
			\textbf{x}=(x_1, x_2,\ldots, x_{t})\sim N(0,\Sigma),\; \text{with}
		\]
		\begin{equation} \label{severaljumps}
			Cov(x_s,x_t)=\left\{\begin{matrix}1&\ldots s=t\\
			0.8\; exp(-r\left |  t-s\right |)&\ldots 0 < |t-s| < 30\\
			0.7\; exp(-r\left |  t-s\right |)&\ldots 30 \leq |t-s| < 73\\
			0.6\; exp(-r\left |  t-s\right |)&\ldots 73 \leq |t-s| < 88\\
			0.5\; exp(-r\left |  t-s\right |)&\ldots 88 \leq |t-s| < \infty \\
			\end{matrix}\right.
		\end{equation}

Graphical representations of the influences are given in Figure \ref{fig:model4}, where the number of jumps is increased by steps.

\begin{figure}[!ht]
\begin{center}
\subfigure[Theoretical autocorrelation function]{\includegraphics[width = 0.45\textwidth, trim = 0 0 0 35,clip]{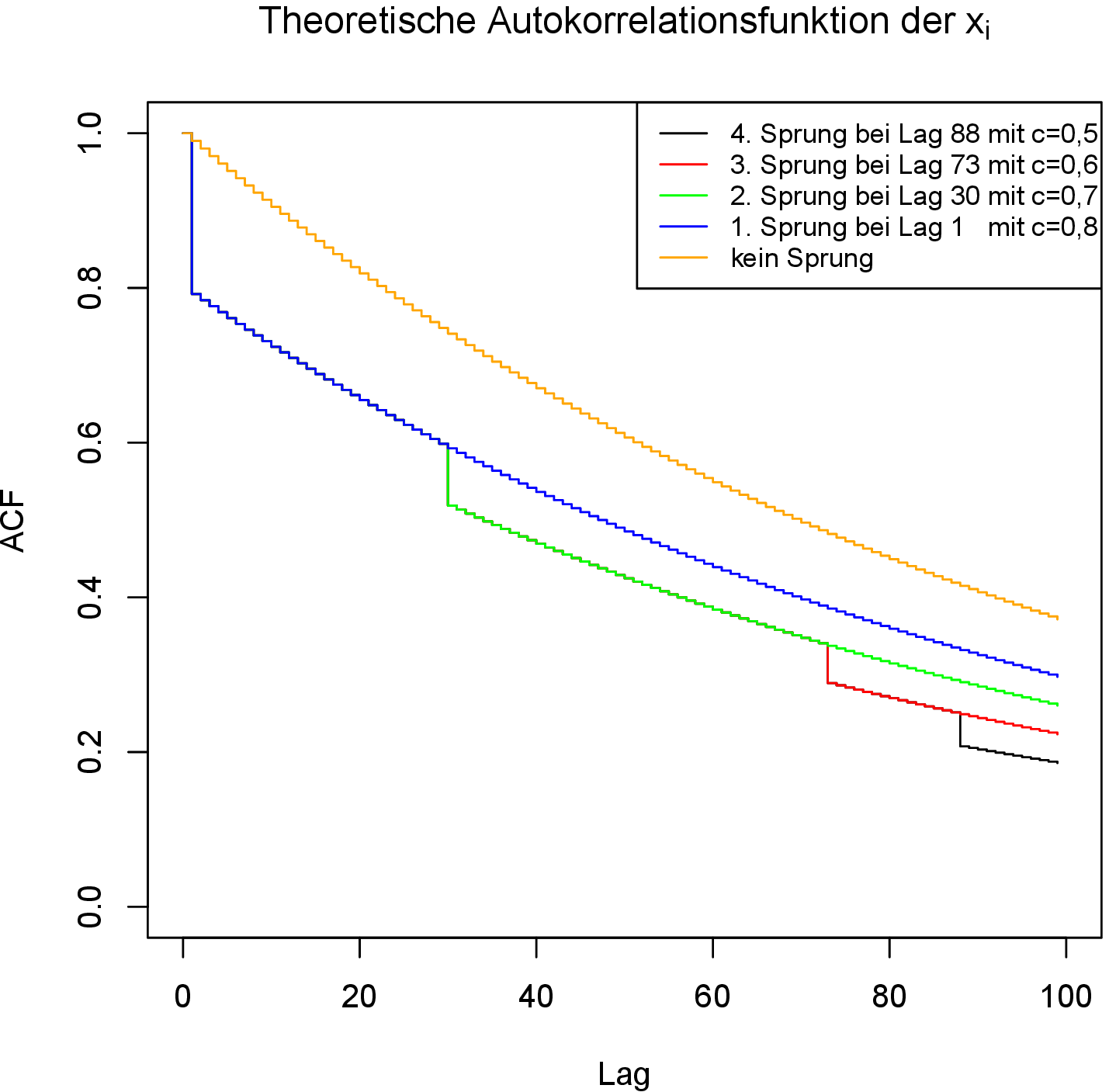}}
\subfigure[Path of $Y_t$ with up to four jumps]{\includegraphics[width = 0.45\textwidth, trim = 0 0 0 35,clip]{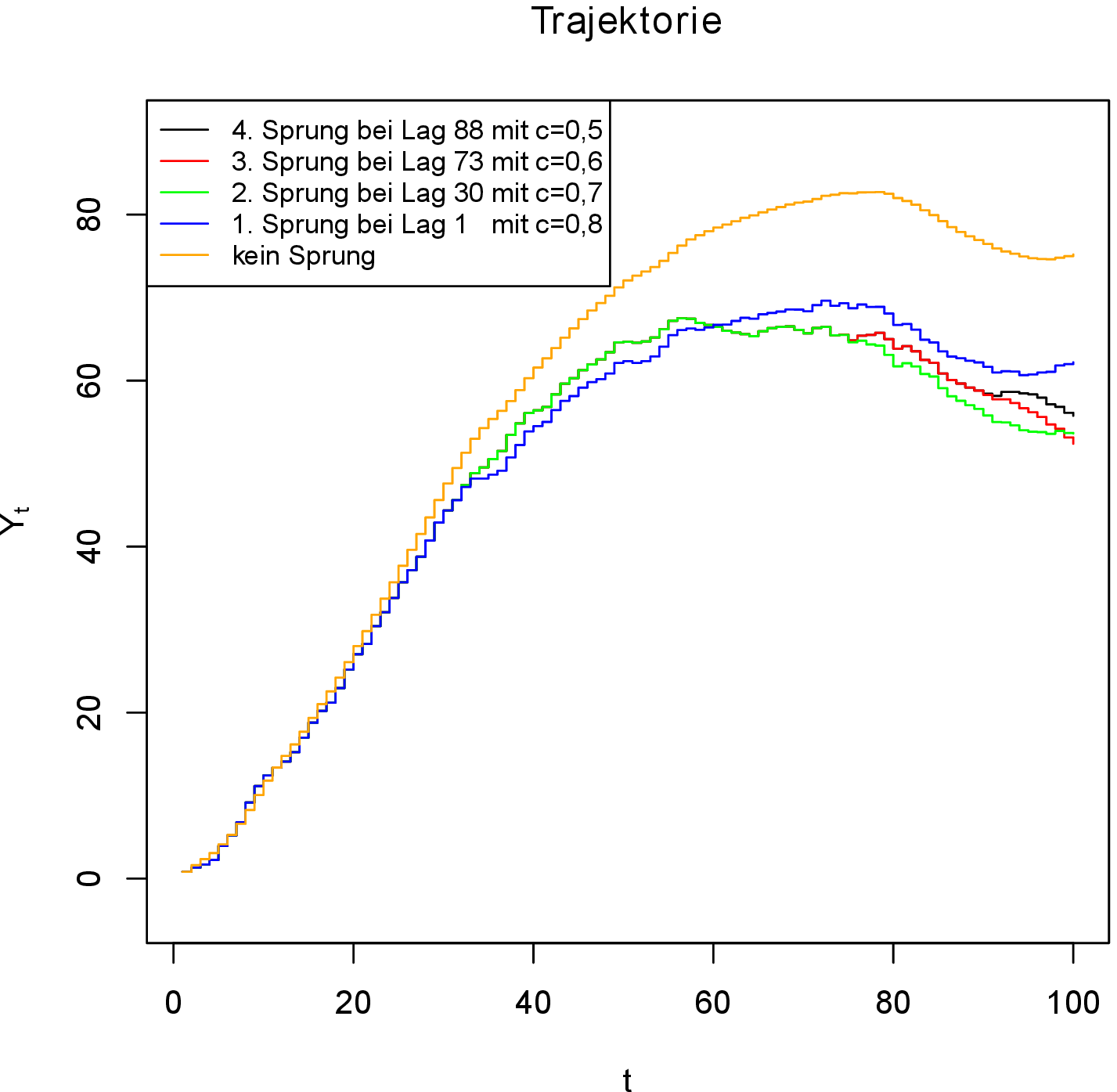}}
\caption{(a) Theoretical autocorrelation function of the $x_i$ and (b) a path of $Y_t$ with up to four jumps with $r=0,01$.}
\label{fig:model4}
\end{center}
\end{figure}

Through the parameters $r$ and $c$ the model is completely specified. To fit this model to a specific data, the parameters $r$ and $c$ have to be estimated.

\subsection{Paths and Simulations for power exponential abc class}
For the next simulation setup covariance structure was chosen as power exponential:
\begin{equation}
Cov(x_s,x_t)=\left\{\begin{matrix}1&\ldots s=t\\
exp^{-\left (|  t-s\right |^p)} &\ldots s\neq t
\end{matrix},\right.
\end{equation}
for different values of $p= 1,2 $ and $10.$
In order to give a graphical overview of the impacts of the parameters on the covariance structure and as a consequence on the trajectory, the paths of $Y_t$ are plotted for different values of the parameters. Therefore the known covariance structure is extended to be
\begin{equation}
Cov(x_s,x_t)=\left\{\begin{matrix}1&\ldots s=t\\
c\;exp^{-r\left (|  t-s\right |)^p} &\ldots s\neq t
\end{matrix}.\right.
\end{equation}
It is now possible to check for the impacts of every parameter on the path of $Y_t$, ceteris paribus. The visualization of the differences can be seen in Figure \ref{fig:Paths}. Main differences can be observed for a value of $r = 0.35$ where the path is shifted upwards. The highest difference in the other direction can be observed for a value of $c = 0.2$, however, changes in the parameters from $r = 1, c = 1$ and $p = 1$ can lead to changes to both directions in the trajectory.
\begin{figure}[!ht]
	\centering
		\includegraphics[width=0.45\textwidth, trim = 0 20 0 0, clip]{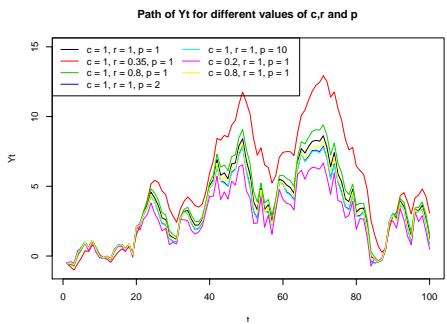}
	\caption{Paths for different values of parameters c,r and p}
	\label{fig:Paths}
\end{figure}

\subsection{Scaling and difference between nugget and later jumps}
As proven by \cite{Crum}, if $C$ is isotropic and positive definite on $R^m,\ m>1$ then $C$ is continuous except perhaps at the origin  $d=0$ (nugget effect). However, for $m=1$ the latter is not true anymore and here we illustrate difference between
both cases, i.e.

\textbf{A)} \emph{Nugget effect} (discontinuity at the origin  $d=0$)

\textbf{B)} \emph{several jumps}, i.e.~covariance function of equation (\ref{severaljumps})

Figure \ref{ABcomparison} shows simulated differences between both cases, A and B for the process $X_i$ itself and its related random walk $Y_t$.
It can be seen that for small values of $r (< 0.1)$ the  differences between both cases are obvious.
The effect of scaling parameter $r$ is obvious, since the differences between A and B are negligible (i.e.~differences around $9\times 10^{-15}$) for e.g.~$r=1.$

\begin{figure}[h]
		\subfigure[Difference between cases A and B for various r]{\includegraphics[width = 0.45\textwidth]{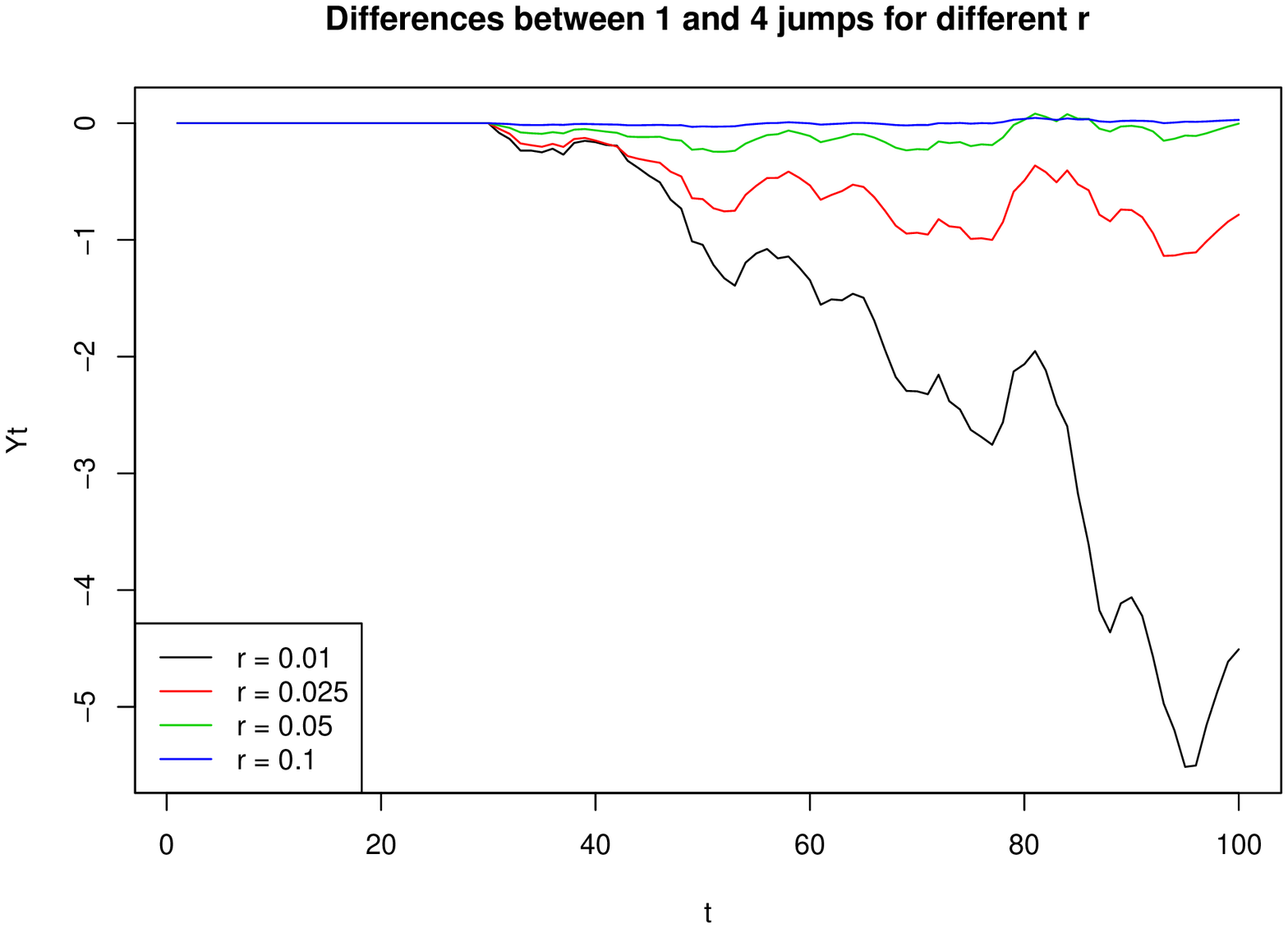}}
		\subfigure[Comparison between cases A and B: random walks for  $r = 0.025.$]{\includegraphics[width = 0.45\textwidth]{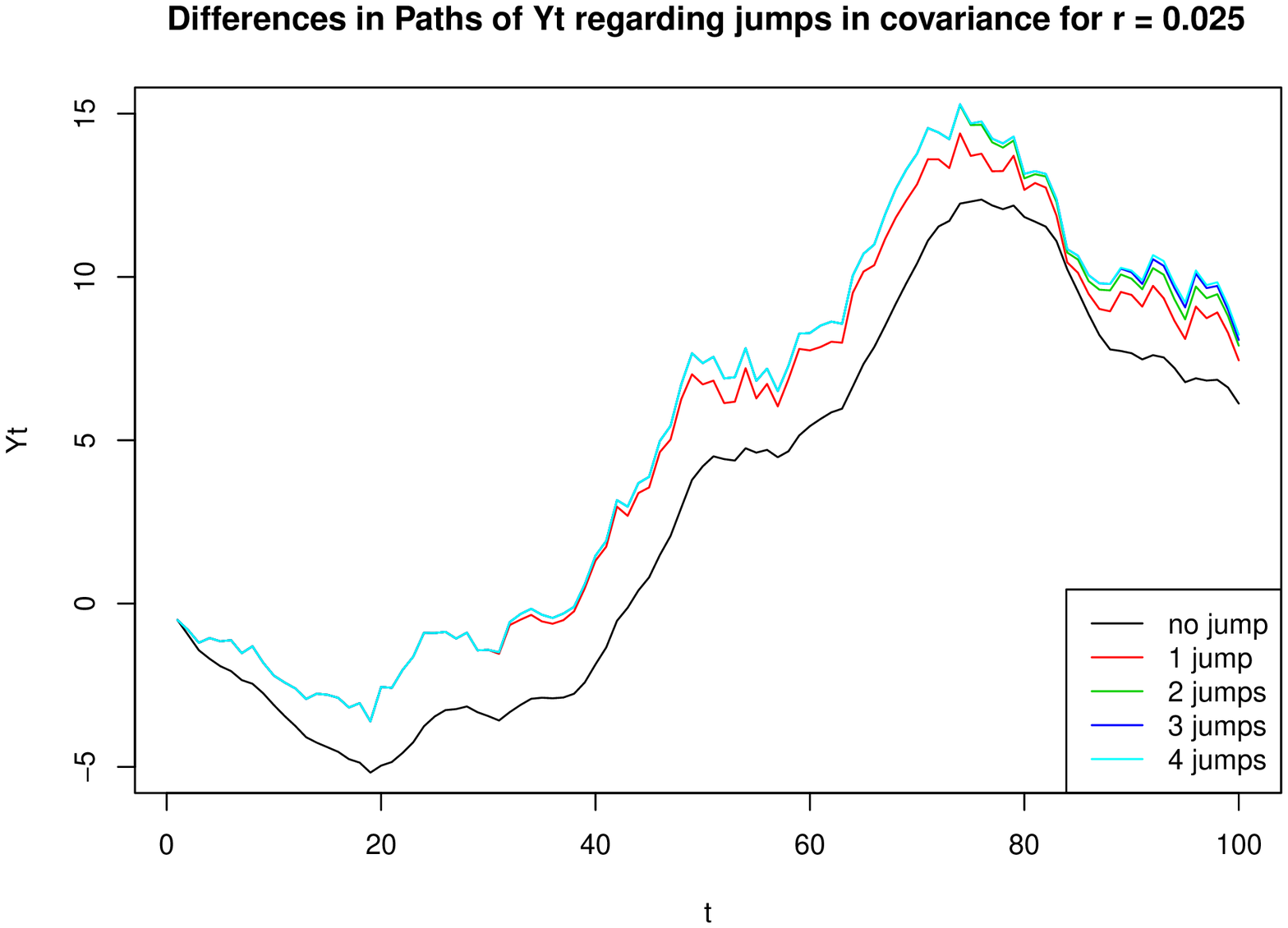}}
	\caption{Comparison between cases A and B}
	\label{ABcomparison}
\end{figure}

\subsection{Convergence of the empirical autocorrelation function}

Up to this point, the empirical autocorrelation function, which can be estimated from simulated data, was disregarded. By plotting both theoretical and empirical autocorrelation functions in one plot, it can be seen that the true correlation between the data is underestimated by the empirical autocorrelation function. Top row of Figure \ref{fig:model5} does not show differences in the empirical covariance function visible with regard to having continuous or no continuous covariance matrix. If the parameter space $T$ is changed to the interval  $[1;6,000]$, it can be seen in bottom row of Figure \ref{fig:model5}, that the empirical autocorrelation function converges to the theoretical autocorrelation function, at least for the first 100 Lags.

\begin{figure}[!ht]
\begin{center}
\includegraphics[width = 0.45\textwidth, trim = 0 0 0 35,clip]{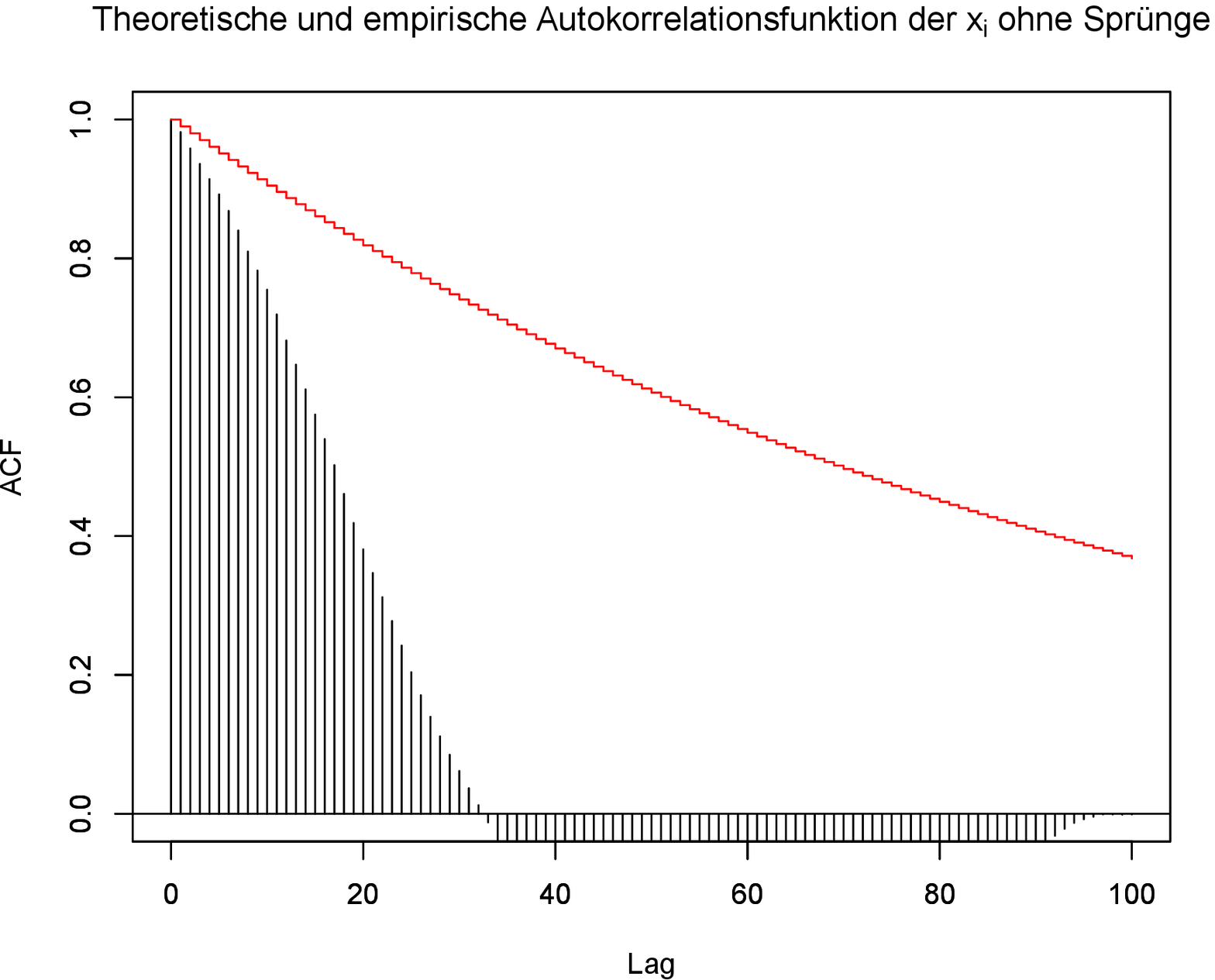}
\includegraphics[width = 0.45\textwidth, trim = 0 0 0 35,clip]{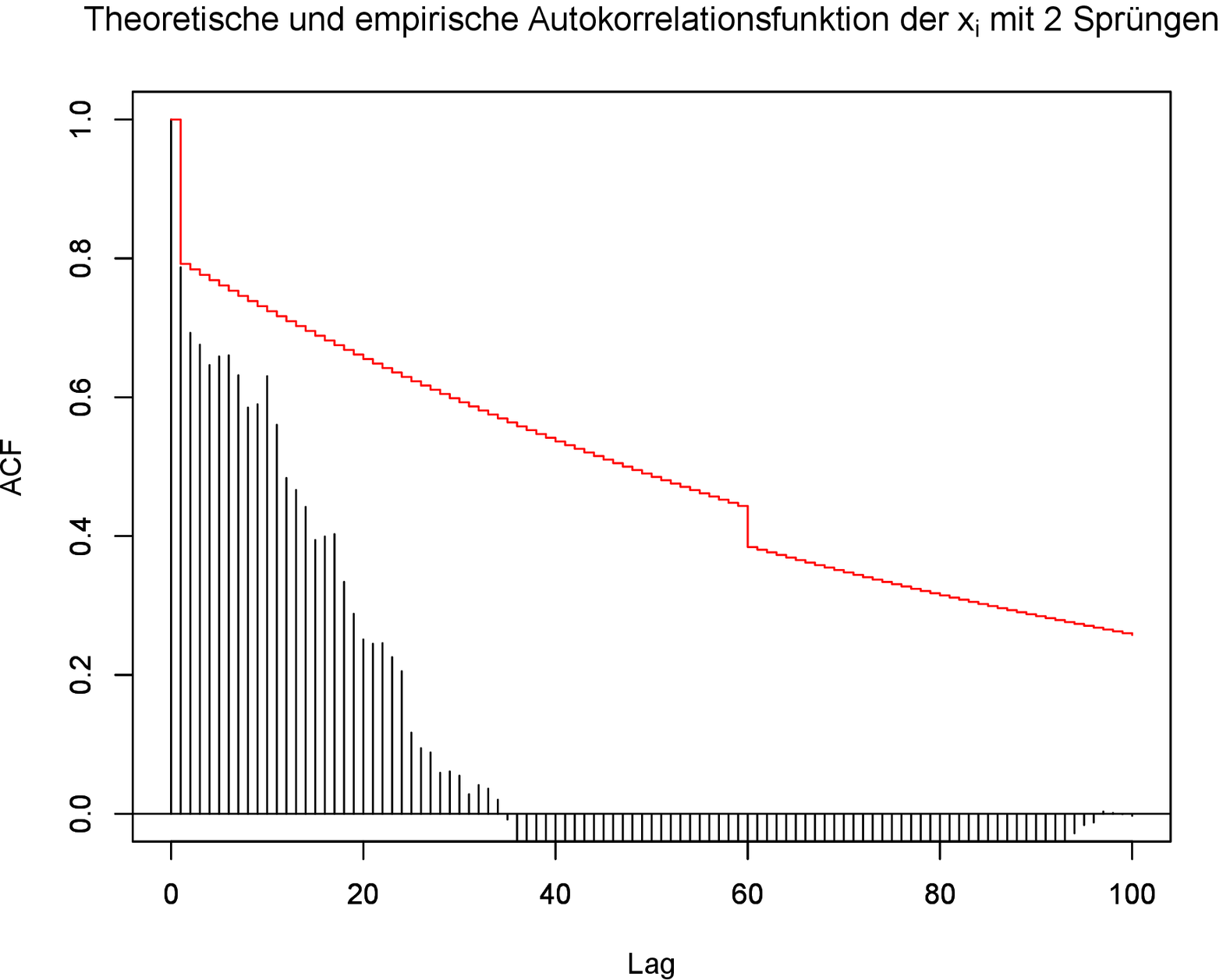}
\includegraphics[width = 0.45\textwidth, trim = 0 0 0 35,clip]{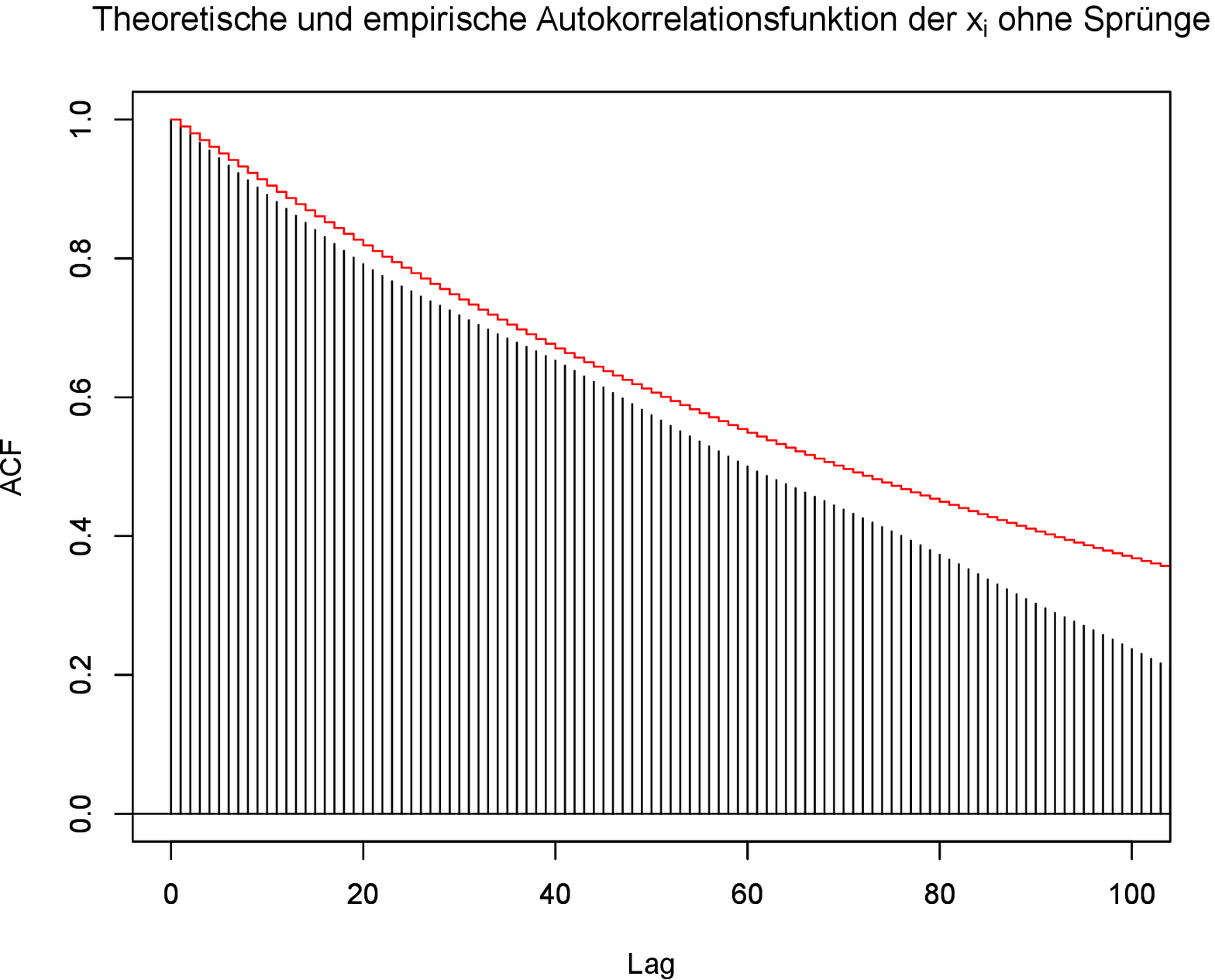}
\includegraphics[width = 0.45\textwidth, trim = 0 0 0 35,clip]{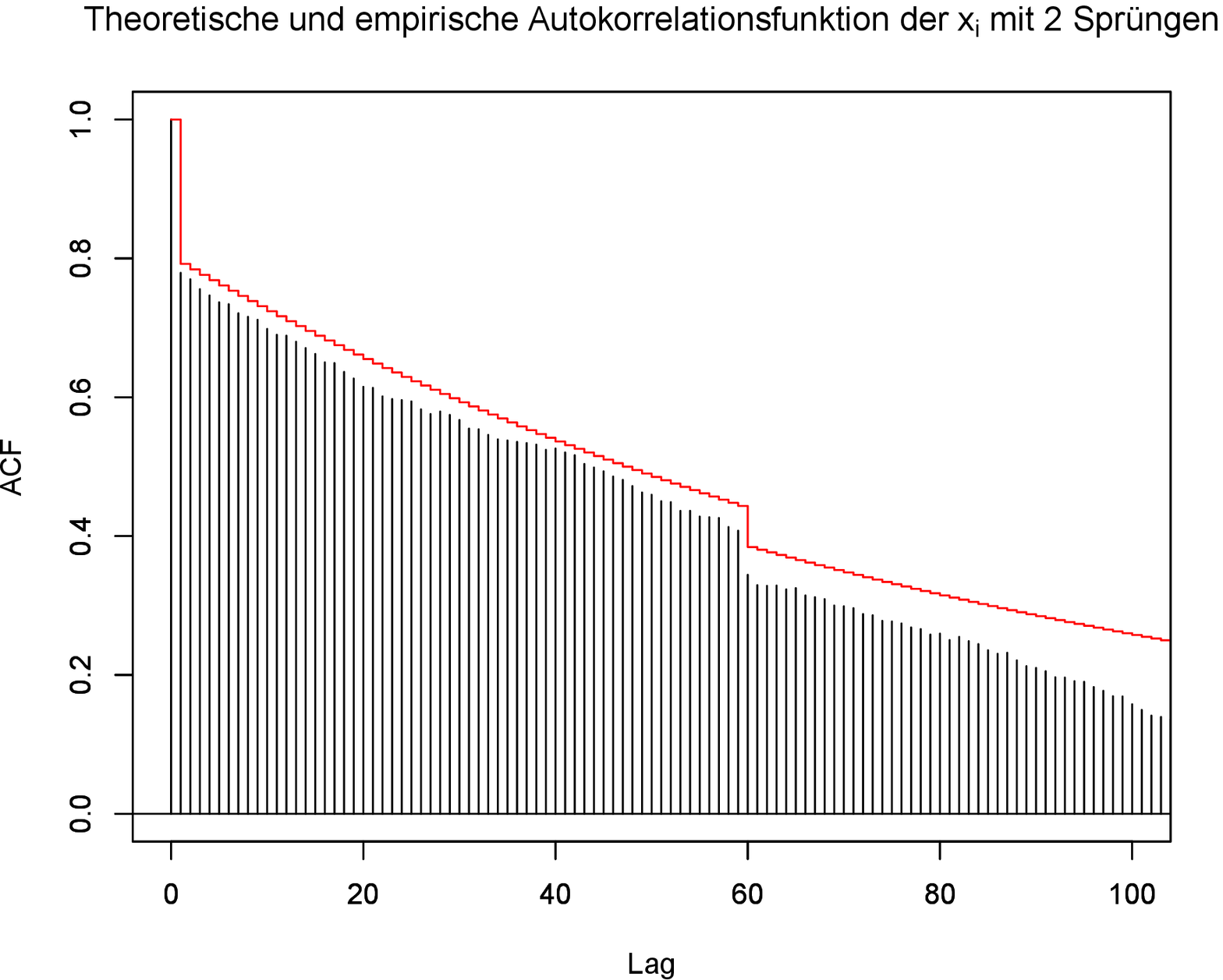}
\caption{Theoretical and empirical autocorrelation function of the $x_i$ with continuous covariance matrix (left plot) and with 2 jumps in the covariance matrix with $r=0.01$ (top right plot). Extending the parameter space to the interval $[1;6,000]$ is performed for the plot bottom right.}
\label{fig:model5}
\end{center}
\end{figure}

The empirical autocorrelation function fails to clearly identify the jumps. Comparing both empirical autocorrelation functions, it can be seen that the data is more scattered when there are jumps. Having a larger sample enables to identify the jumps in the covariance matrix and the convergence of the empirical to the theoretical autocorrelation function is properly visible for the first 100 lags. Considering that in practice, small samples are used to analyze time series, it is supposed that it is impossible to assume that there are jumps in the correlation matrix, although jumps may be well present. This fact confirms the relevance of the model, because it supposes that the data arises from a distribution with discontinuous covariance matrix. Only, if there is no proof of jumps, a continuous covariance matrix may be  accepted.

\subsection{Test for continuity of the covariance matrix}

Differences between the theoretical and the empirical autocorrelation functions can be used to analyze whether data arises from a process with or without jumps in the covariance. Hence, the differences are calculated at every lag and its absolute value is summed up. This sum is called sum of residuals $T$ and defined as
\begin{equation} 	T = \sum_{L=0}^{n-1}|\rho(L)-\hat{\rho}(L)|. \end{equation}
%
%

Calculating the sum of residuals from data of Figure \ref{fig:model5} leads to a value of 63.4 for the data with continuous covariance matrix (top left) and 48.96 for the data with two jumps in the covariance matrix (top right). Data from plots in the second row lead to sums of residuals of 352.62 and 271.54, respectively. The sum of residuals is smaller if there are jumps in the covariance matrix. Therefore, we analyzed if there is a general difference in the parameter space of $[1;100]$. To determine whether the sum of residuals is smaller when there are jumps in the covariance matrix, various samples, with fixed $r$ and sample size, were generated. The jump points and jump heights $c$ are varied in a way, that the covariance matrix remains positive semi-definite. Table \ref{tab1} shows the results and confirm the assumption.

\begin{table}[!ht]
    \resizebox{0.99\textwidth}{!}{\begin{tabular}{|l|l|l|l|l|l|l|l|l|l|}
        \toprule
 \multicolumn{9}{|c|}{1 jump at Lag 1 with various $c$} \\
    \multicolumn{3}{|c|}{$c$}   & 0     & 0.2    & 0.4    & 0.6    & 0.8    & 1      \\ \midrule
        \multicolumn{3}{|c|}{sum of residuals} & 4.74 & 13.38 & 25.51 & 38.02 & 50.52 & 63.03 \\         \midrule\midrule
		\multicolumn{9}{|c|}{2 jumps, 1.~ with $c = 0.8$ and 2.~jump at Lag s with c = 0.7}\\
			\multicolumn{1}{|c|}{$s$}   & 30    & 40    & 50    & 60    & 70   & 80   & 90   & 99   \\ \hline
        sum of residuals & 46.78 & 47.48 & 48.12 & 48.71 & 49.23 & 49.70 & 50.13 & 50.49   \\
        \midrule\midrule
   \multicolumn{9}{|c|}{3 jumps, 1-2. as above and 3.~jump at Lag $s$ with c  = 0.6} \\
\multicolumn{5}{|c|}{$s$}     & 73     & 75     & 85     & 98        \\ \hline
        \multicolumn{5}{|c|}{sum of residuals} & 45.63 & 45.73 & 46.18 & 46.70    \\
        \midrule\midrule
  \multicolumn{9}{|c|}{ 4 jumps, 1-3. jumps as above, 4.~jump at Lag 88 (c = 0.5)} \\\midrule
	\multicolumn{9}{|c|}{sum of residuals: 45.18}   \\
        \bottomrule\bottomrule
    \end{tabular}}
		\caption{$r$ = 0.01, $n$ = 100. Comparison with respect to sum of residuals of 1 to 4 jumps in the correlation structure given parameters of $c$ and $s$.}
				\label{tab1}
\end{table}

The smaller the jump height $c$, the more decreases the sum of residuals. Additionally, the lag, i.e.~when the jump takes place, influences the sum of residuals. As it is impossible to simulate every combination from $r$, number of jumps, jump height and jump point, following restrictions are made to perform hypothesis testing for continuous covariance matrix: there is only one jump at Lag 1 in the covariance matrix. In addition $r$ and $n$ are fixed.  Assuming, that $c$ arises from an unknown distribution function $F_c$, 10,000 different paths are simulated and their sum of residuals is calculated. $c$ is in this case a random number from $F_c$. The sums of residuals are illustrated with a histogram and scatter plot afterwards. Figure \ref{fig:model8} shows the results for various distribution functions $F_c$.

\begin{figure}[!ht]
\begin{center}
\subfigure[U(0,1)]{\includegraphics[width = 0.23\textwidth, trim = 0 0 0 38,clip]{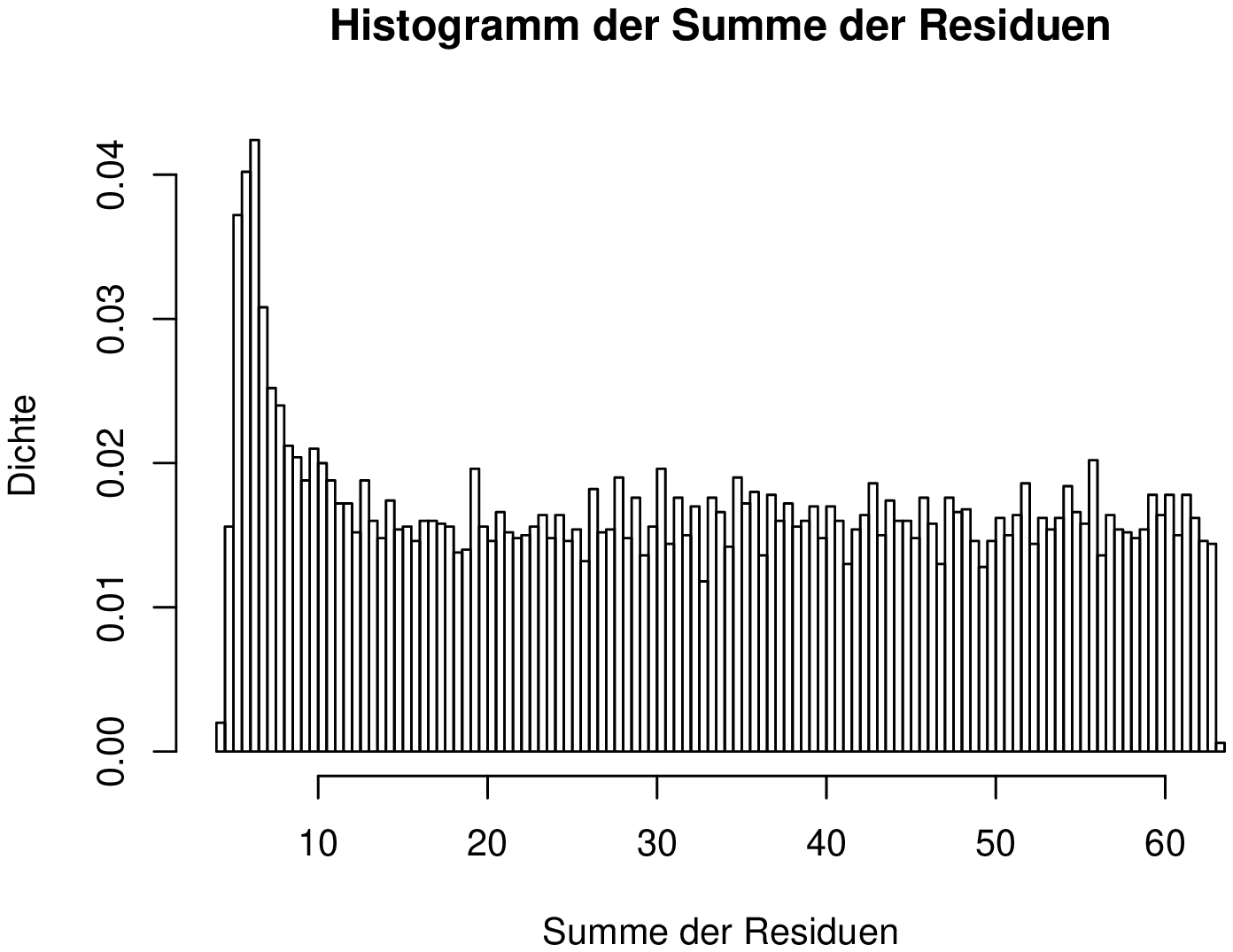}}
\subfigure[U(0,1)]{\includegraphics[width = 0.23\textwidth, trim = 0 0 0 38,clip]{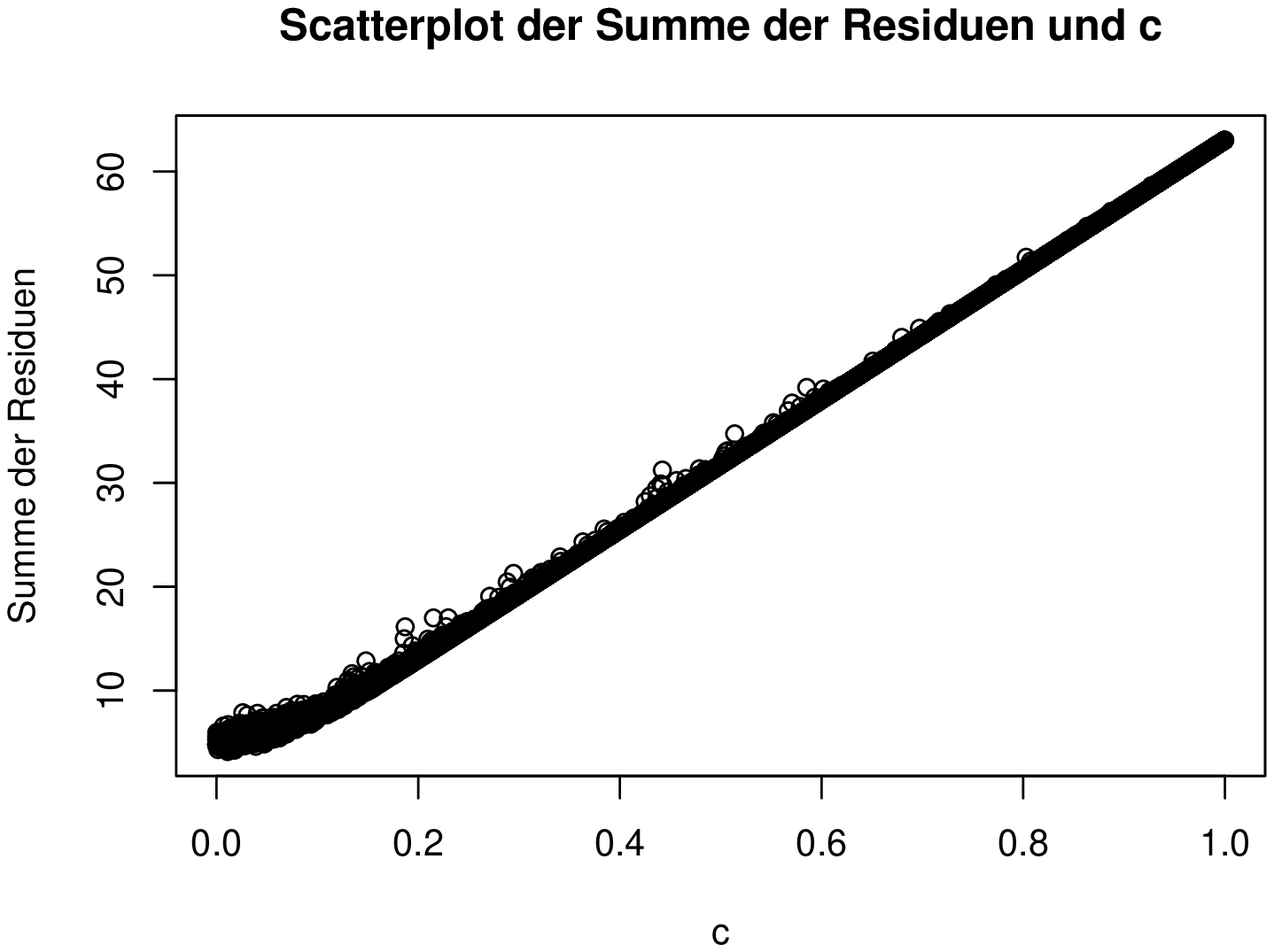}}
\subfigure[Gamma(2,1)]{\includegraphics[width = 0.23\textwidth, trim = 0 0 0 38,clip]{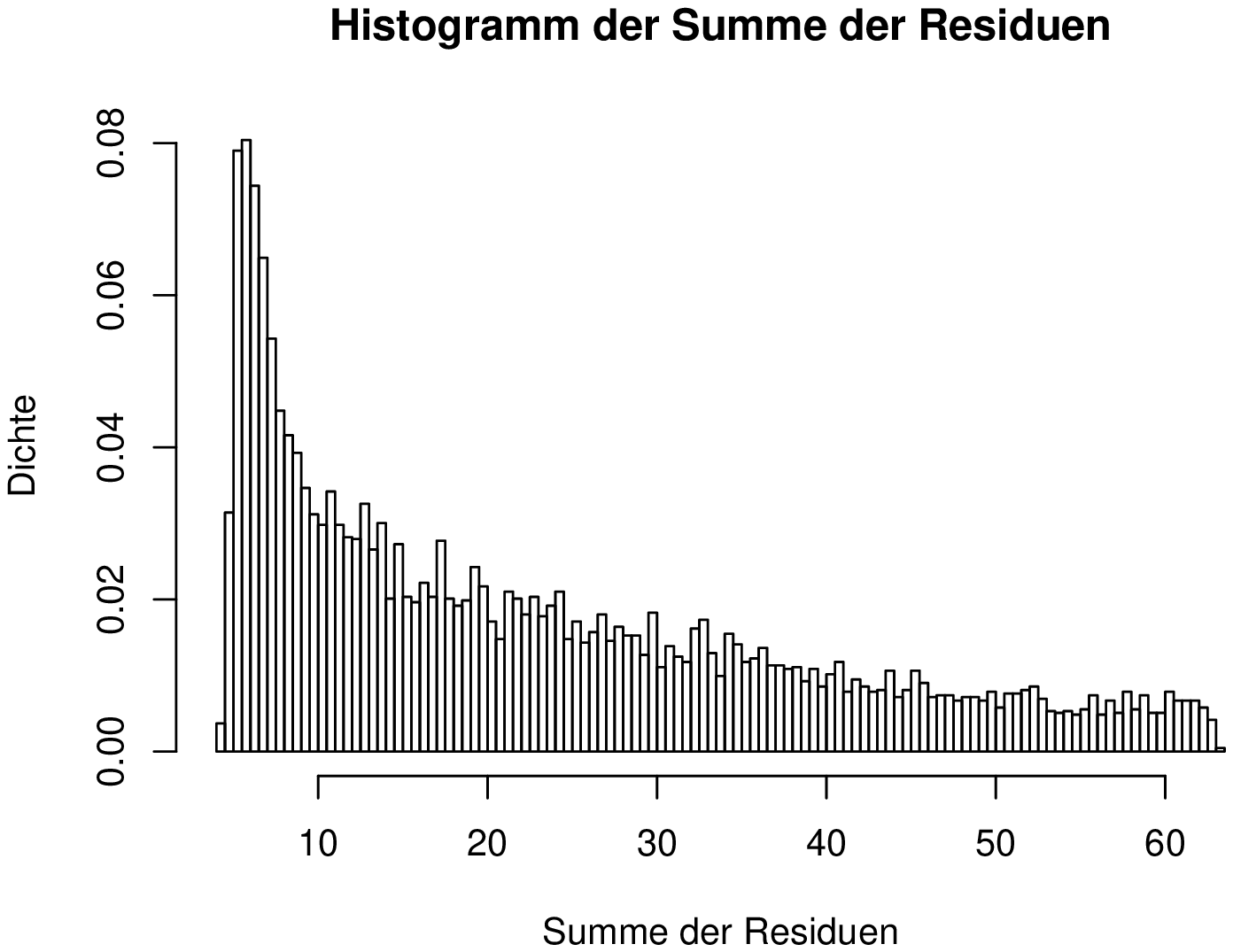}}
\subfigure[Gamma(2,1)]{\includegraphics[width = 0.23\textwidth, trim = 0 0 0 38,clip]{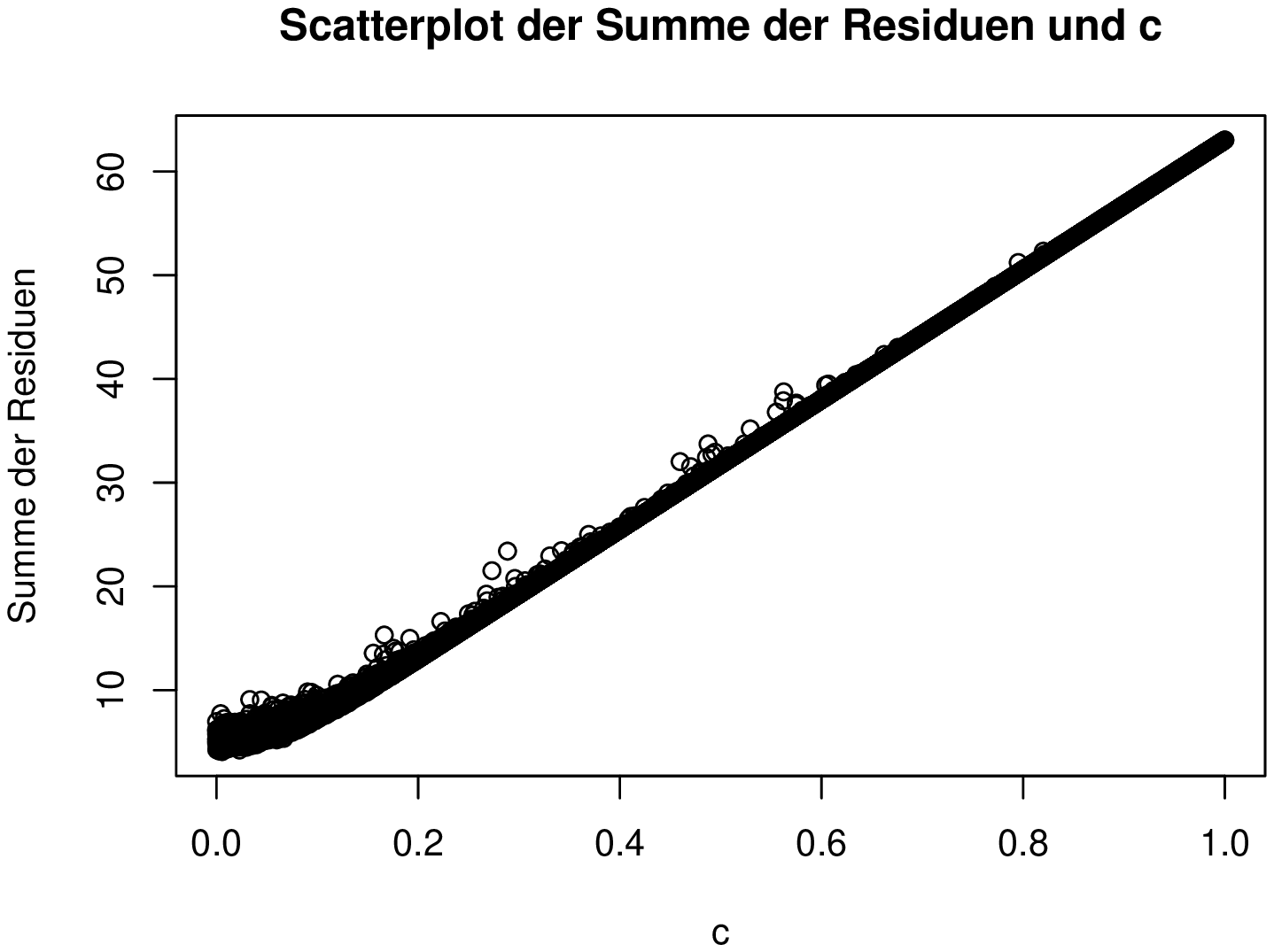}}
\subfigure[Poisson(5)/10]{\includegraphics[width = 0.23\textwidth, trim = 0 0 0 38,clip]{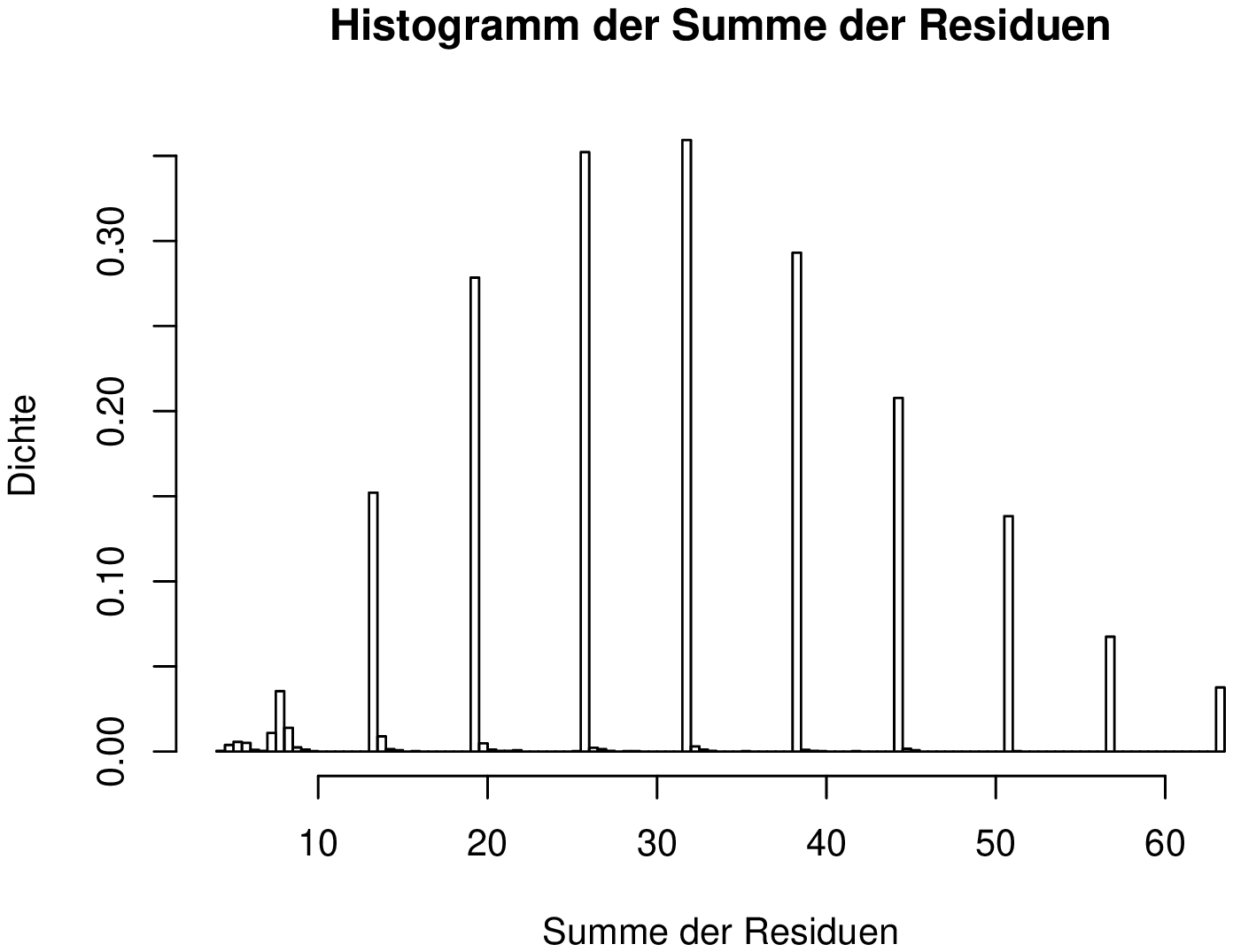}}
\subfigure[Poisson(5)/10]{\includegraphics[width = 0.23\textwidth, trim = 0 0 0 38,clip]{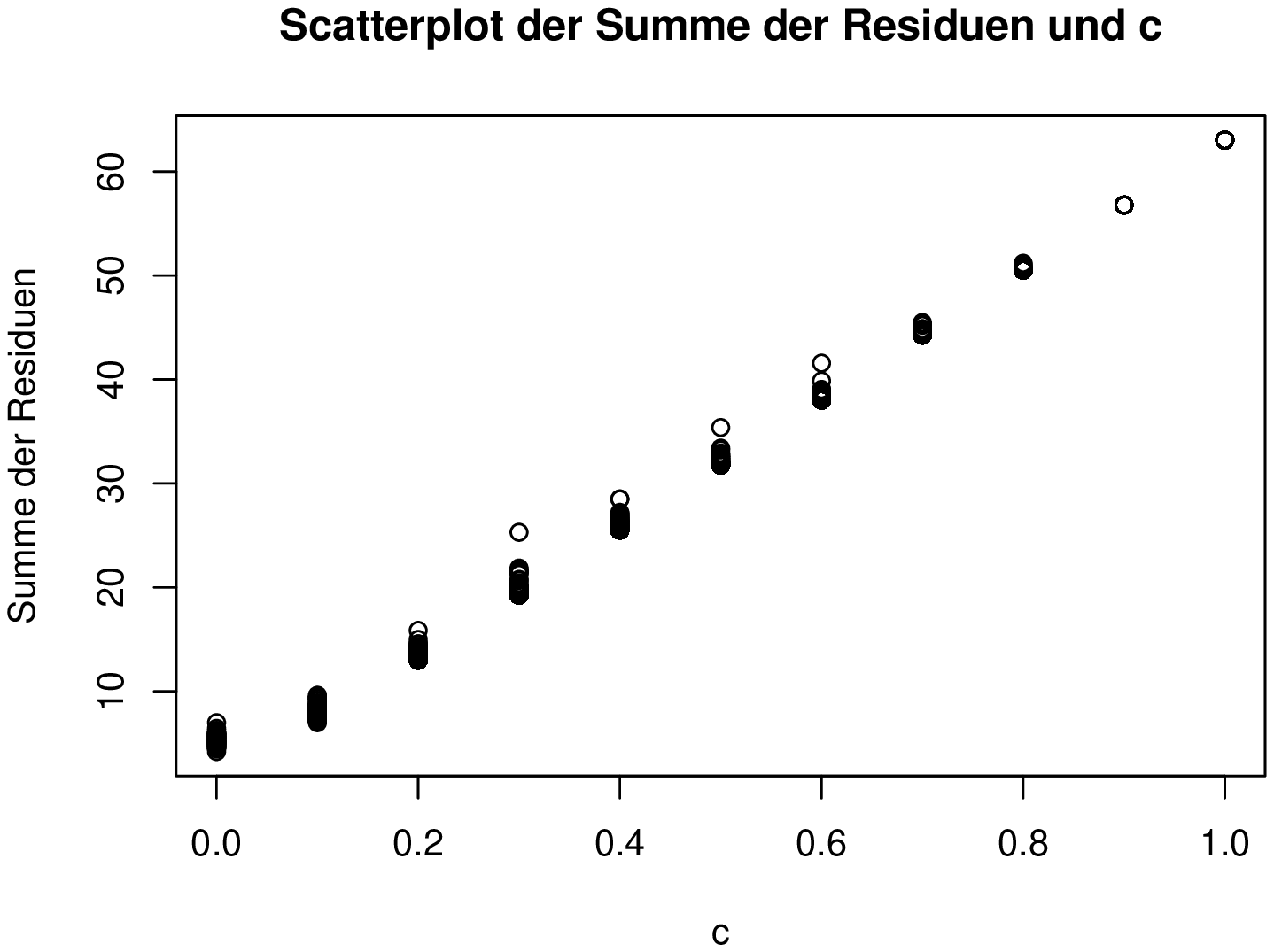}}
\subfigure[Bin(20,0.5)/20]{\includegraphics[width = 0.23\textwidth, trim = 0 0 0 38,clip]{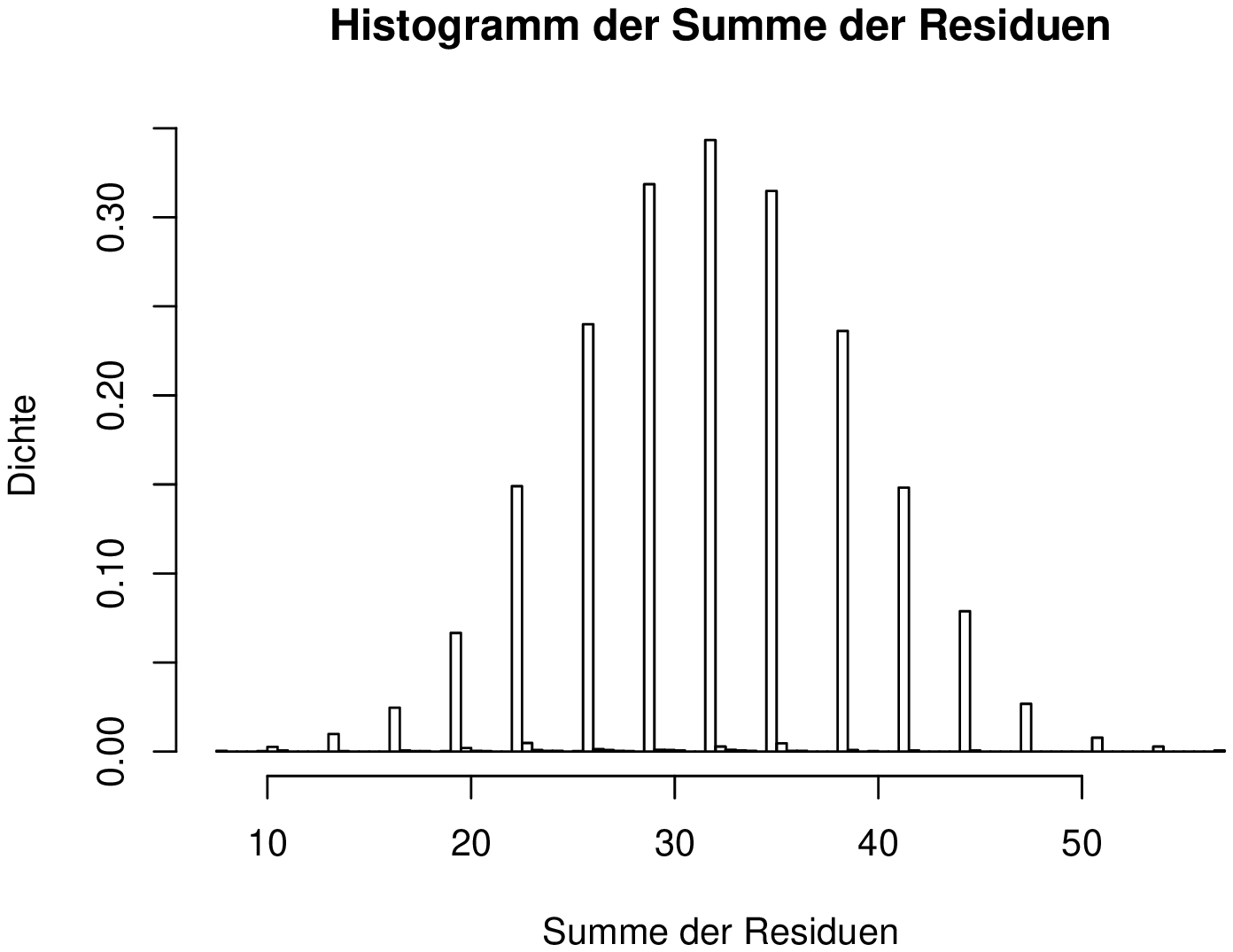}}
\subfigure[Bin(20,0.5)/20]{\includegraphics[width = 0.23\textwidth, trim = 0 0 0 38,clip]{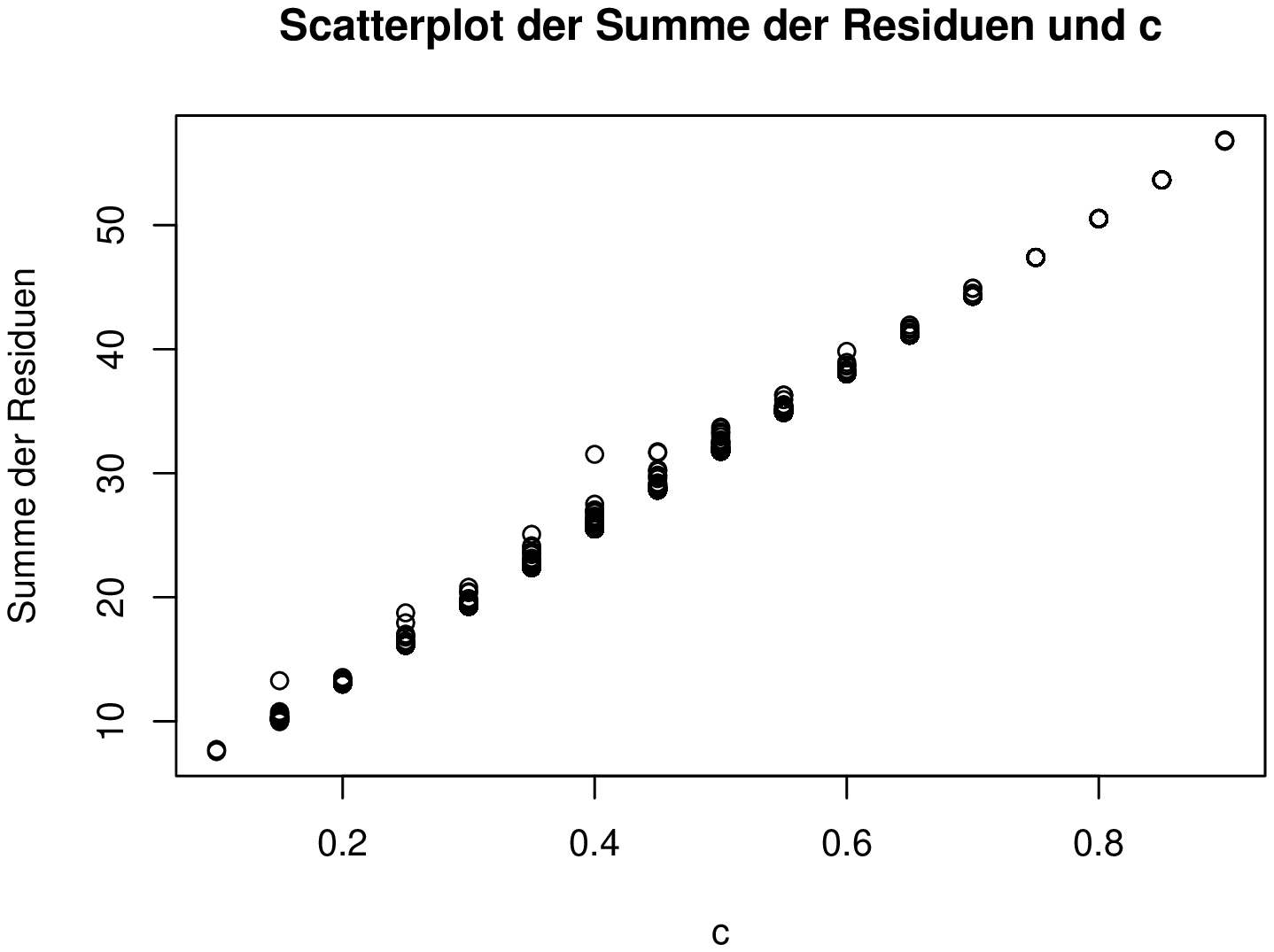}}
\caption{Sum of residuals with $r=0.01$, $n=100$ and $c$ distributed (a)(b) Uniform ($c\sim U(0;1)$), (c)(d) Gamma ($c\sim \Gamma(2;1)$), (e)(f) Poisson ($c\sim \frac{P(5)}{10}$) and (g)(h) Binomial ($c \sim \frac{Bin(20;0,5)}{20}$).}
\label{fig:model8}
\end{center}
\end{figure}

When using Gamma or Poisson distribution values of $c$ greater than one were neglected. The fact, that $c$ arises from a truncated distribution function, was disregarded. Figure \ref{fig:model8} shows, that if the distribution function $F_c$ is known, there is a correlation to the distribution function of the sum of residuals. The scatter plots show a linear correlation between the jump heights $c$ and the sum of residuals. The core statement from Figure \ref{fig:model8} is that if we know the distribution function $F_c$ , then the sum of residuals $T$ can be used as a hypothesis test for a continuous covariance matrix.

In the following, it is shown how the histogram of the sum of residuals develops when there are two or three jumps in the correlation matrix. Again, $r$, $n$ and the jump points are fixed. To simplify the problem, it is assumed that $c$ arises from a Uniform distribution. According to the definition: $c_1$ is in the interval [0;1],$c_2$ must be in the interval [0;$c_1$] and $c_3$ in the interval [0;$c_2$], respectively. The simulations are illustrated in Figure \ref{fig:model9}.

\begin{figure}[!ht]
\begin{center}
\includegraphics[width = 0.45\textwidth, trim = 0 0 0 39,clip]{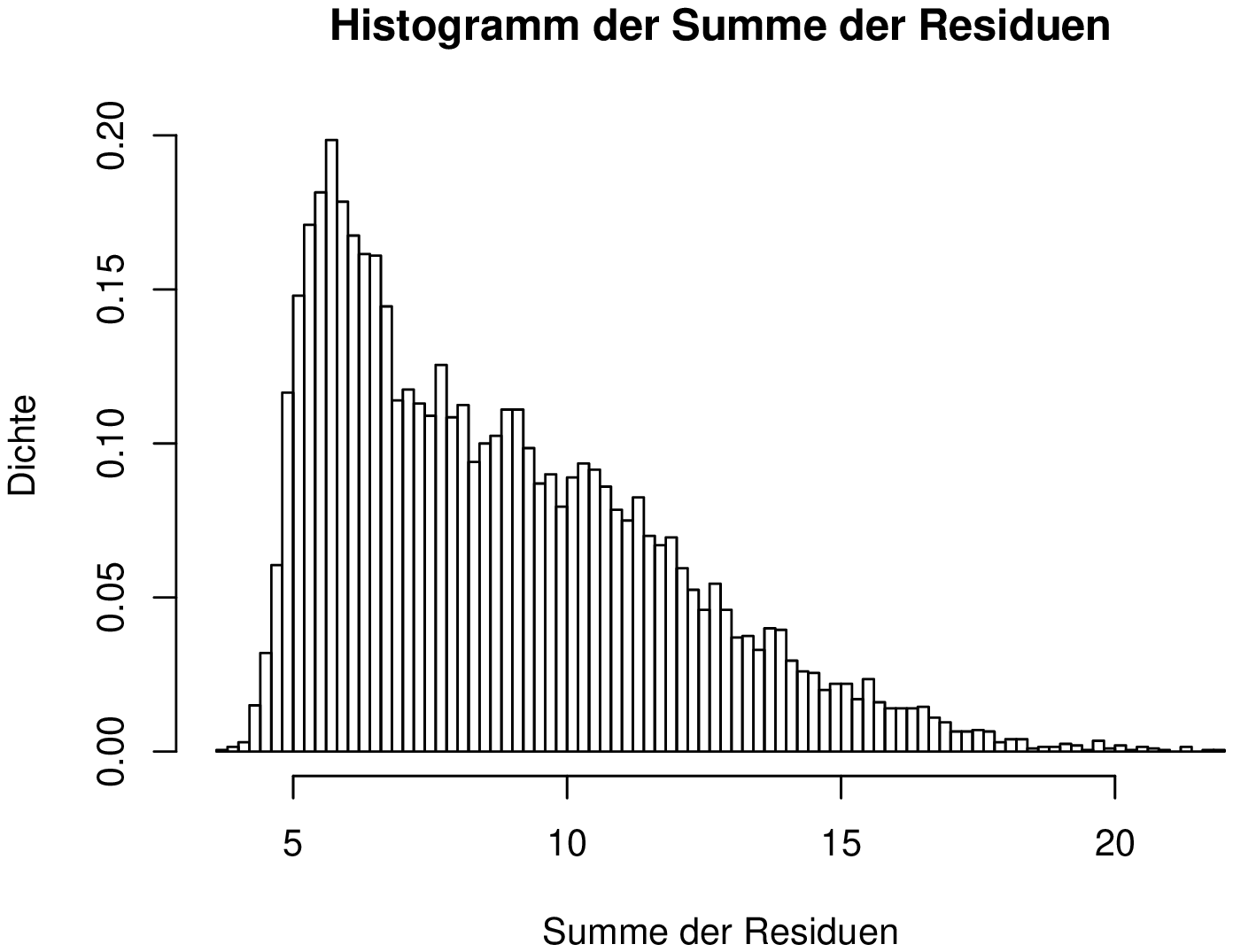}
\includegraphics[width = 0.45\textwidth, trim = 0 0 0 38,clip]{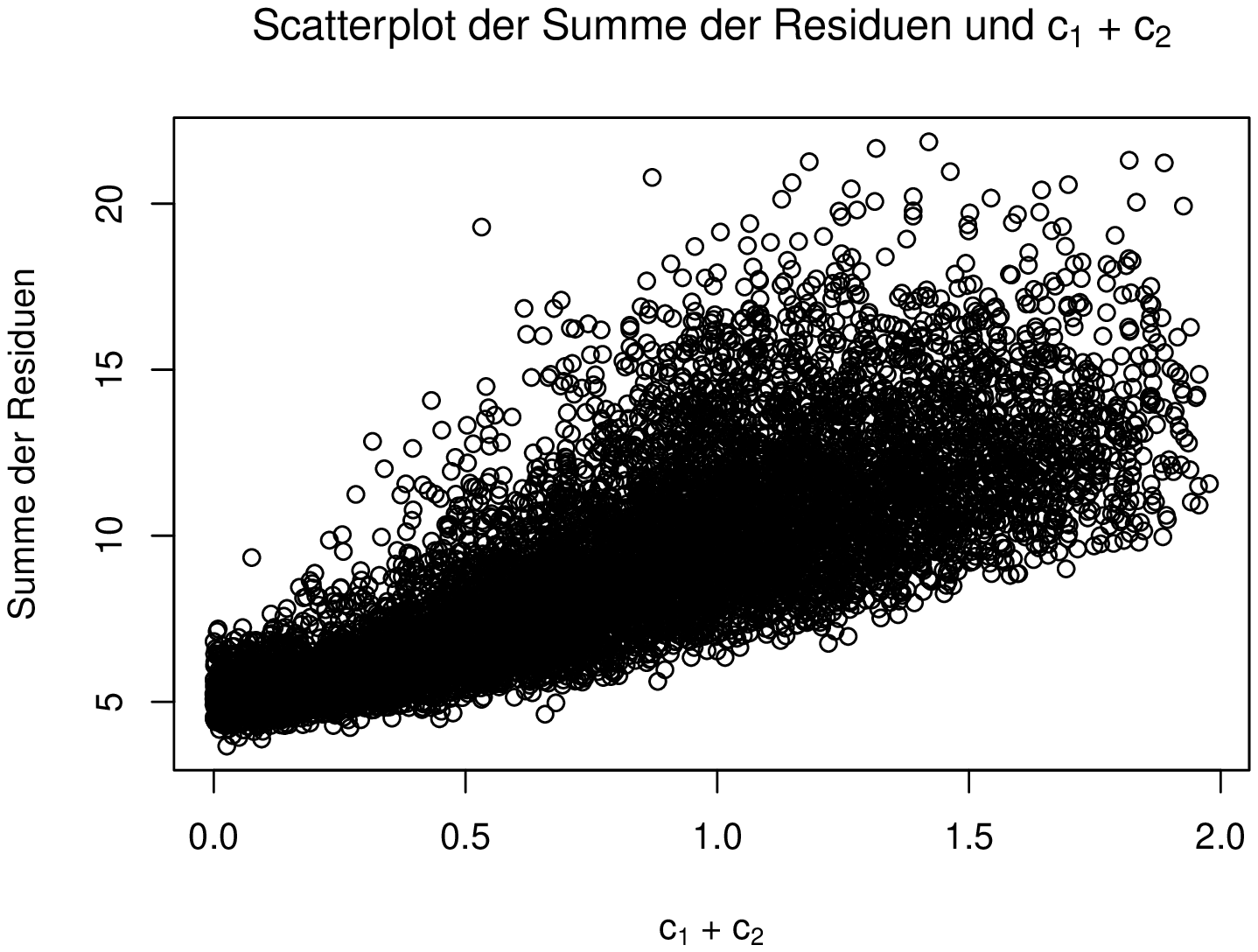}
\includegraphics[width = 0.45\textwidth, trim = 0 0 0 39,clip]{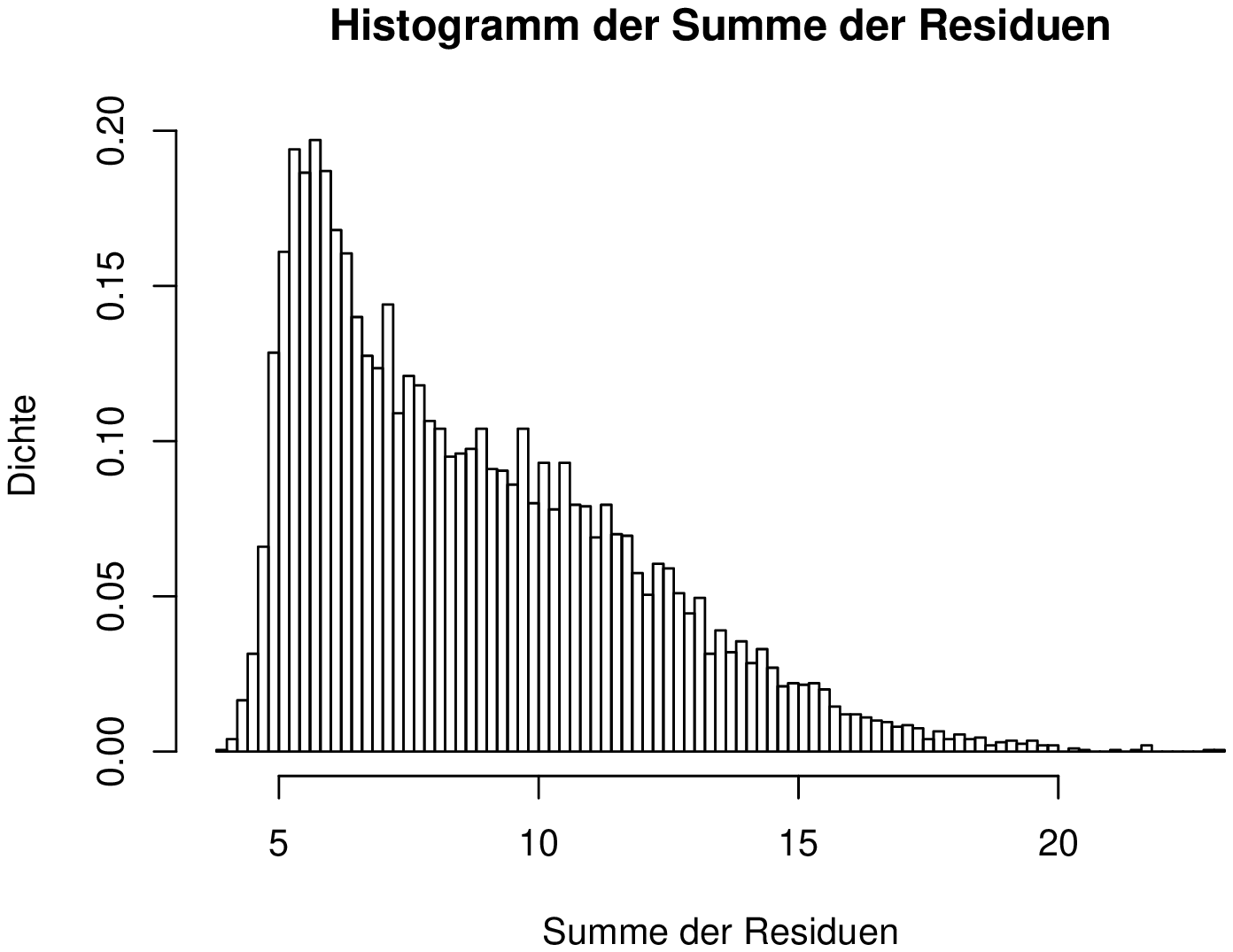}
\includegraphics[width = 0.45\textwidth, trim = 0 0 0 38,clip]{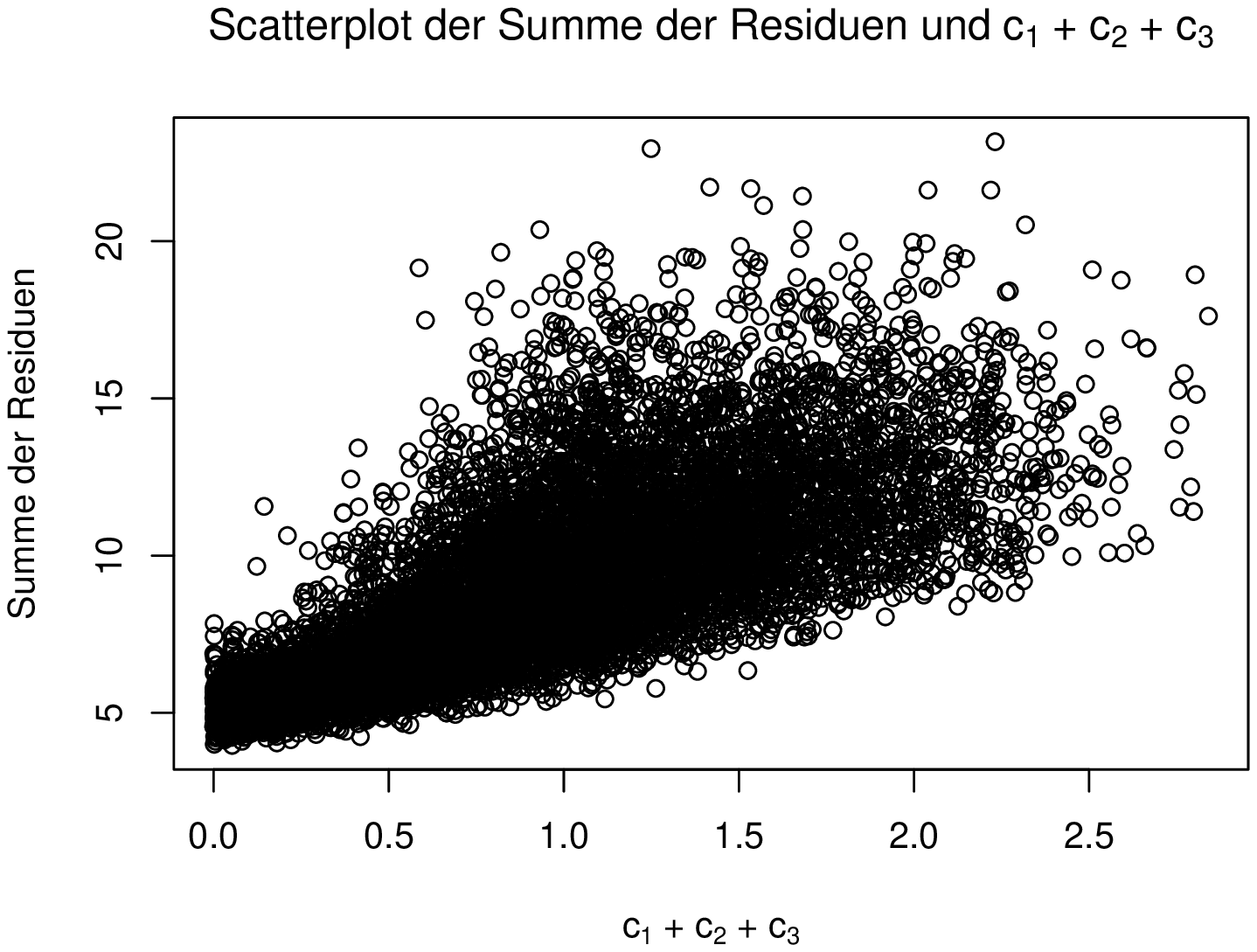}
\caption{Sum of residuals with $r=0.1$ and $n=100$ with 2 jumps ($c_1\sim U(0;1);\;c_2\sim U(0;c_1)$) (top row) and 3 jumps ($c_1\sim U(0;1);\;c_2\sim U(0;c_1);\;c_3\sim U(0;c_2)$) (bottomm row).}
\label{fig:model9}
\end{center}
\end{figure}

Contrary to the case, where covariance matrix consisted of one jump, it is impossible to derive the distribution function of $T$ when there are two or three jumps. Nevertheless, a correlation between the jump height and the sum of residuals can be assumed through the scatter plot.

\subsection{Forecast of stock markets}

A simulation experiment was conducted in order to show how forecast can be performed. We consider that data from Figure \ref{fig:anw1} shows the stock price of the last 90 days of a company (black line). The aim is to predict the stock price of the next ten days.

\begin{figure}[!ht]
\begin{center}
\subfigure[95\% confidence interval for the forecast of the stock price]{\includegraphics[width = 0.45\textwidth, trim = 0 0 0 38,clip]{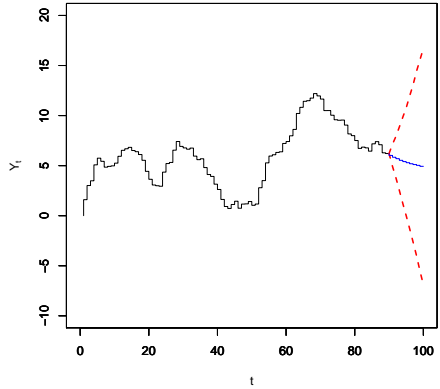}}
\caption{Fictitious stock price of the last 90 days of a company with 95\% confidence interval for its forecast}
\label{fig:anw1}
\end{center}
\end{figure}

Taking first differences leads to the $x_i$. Normally, $x_i$ are used to estimate the parameters $r$ and $c$ and to test whether the covariance matrix is continuous or not. Since this is a fictitious example and the data is simulated, the parameters are known. There are two jumps in the covariance matrix with $r=0.1$. The first jump is at lag 1 with $c_1=0.9$ and the second jump is at lag 30 with $c_2=0.8$. So the covariance matrix is discontinuous in this example.

The forecast for the future stock price is determined by dynamic simulation. First, the correlated $x_i$ are transformed into standard normal distributed random variables by $z_i = \textbf{A}^{-1}\cdot(x_1,...,x_{90}) '$ such that $\textbf{AA} ' = \Sigma.$
The matrix \textbf{\textit{A}} corresponds to the lower triangular matrix of the Cholesky-Decomposition. For one random prediction, the vector z will be completed by 10 standard normal random numbers and back transformed to correlated random variables
by $x_i^{\ast} = \textbf{A}^{\ast}\cdot(z_1,...,z_{90},z_{91}^{\ast},...,z_{100}^{\ast}) '.$
The covariance matrix has to be used in the right dimension. Summing up the vector cumulatively, leads to the original time series including the ten predicted values. If this procedure is repeated 10,000 times, it is possible to create a $(1-\alpha)$- confidence interval like illustrated in Figure \ref{fig:anw1} (blue as forecast and red dashed lines as confidence bounds).

\subsection{Application on Probability of Ruin}
As soon as insurance companies are unable to pay the claims for damages of the assured person, one defines this as ruin of the insurance company.  For prediction of probability of ruin, Cram\'{e}r-Lundberg-model or an alternative approach by \cite{Yuen2001} can be used.  We refer to  \cite{Yuen2001} for more details on the latter  model.

To be able to calculate probabilities of ruin with the presented model, a process is formulated, which indicates the available capital of an insurance company. In order to the collective risk model, the process is called surplus process $U_t$. For $t \in \mathbb{N}_0$ and $U(0)=u$:

\begin{equation} 	U_t=u+\sum_{i=1}^{t}x_{i},\; \text{for}\; t>0, \end{equation}	
\[
	\textbf{x}=(x_1, x_2,\ldots, x_{t})\sim N(0,\Sigma),\; \text{with}
\]
\begin{equation}
	Cov(x_s,x_t)= c\; \text{e}^ {-r |  t-s |}\;\;,
\end{equation}

where $u$ stands for the initial surplus of the insurer at $t=0$ and $x_i$ stands for the profit or loss of the insurance company at point $i$. An insurer generates a profit, if it gets more money by rates than it has to spend at point $i$. In the other way around the insurer has a loss of money while considering, that the insurer provides competitive rates, $E(x_i)=0$. The probability of ruin $\psi(u,t)$ is calculated in the same way like \cite{Yuen2001} showed in their paper:

\begin{equation}	\psi(u,t)=1-Pr(U_j\geq0;j=0,1,2,...,t) \end{equation}


A simulated experiment is conducted in order to show how the forecast is done. Figure \ref{fig:anw3}(a) shows the surplus process $U_t$ of the last 90 weeks of an insurance company.

\begin{figure}[!ht]
\begin{center}
\subfigure[Fictitious surplus process $U_t$ of the last 90 weeks of an insurance company]{\includegraphics[width=0.45\textwidth, trim = 0 0 0 38,clip]{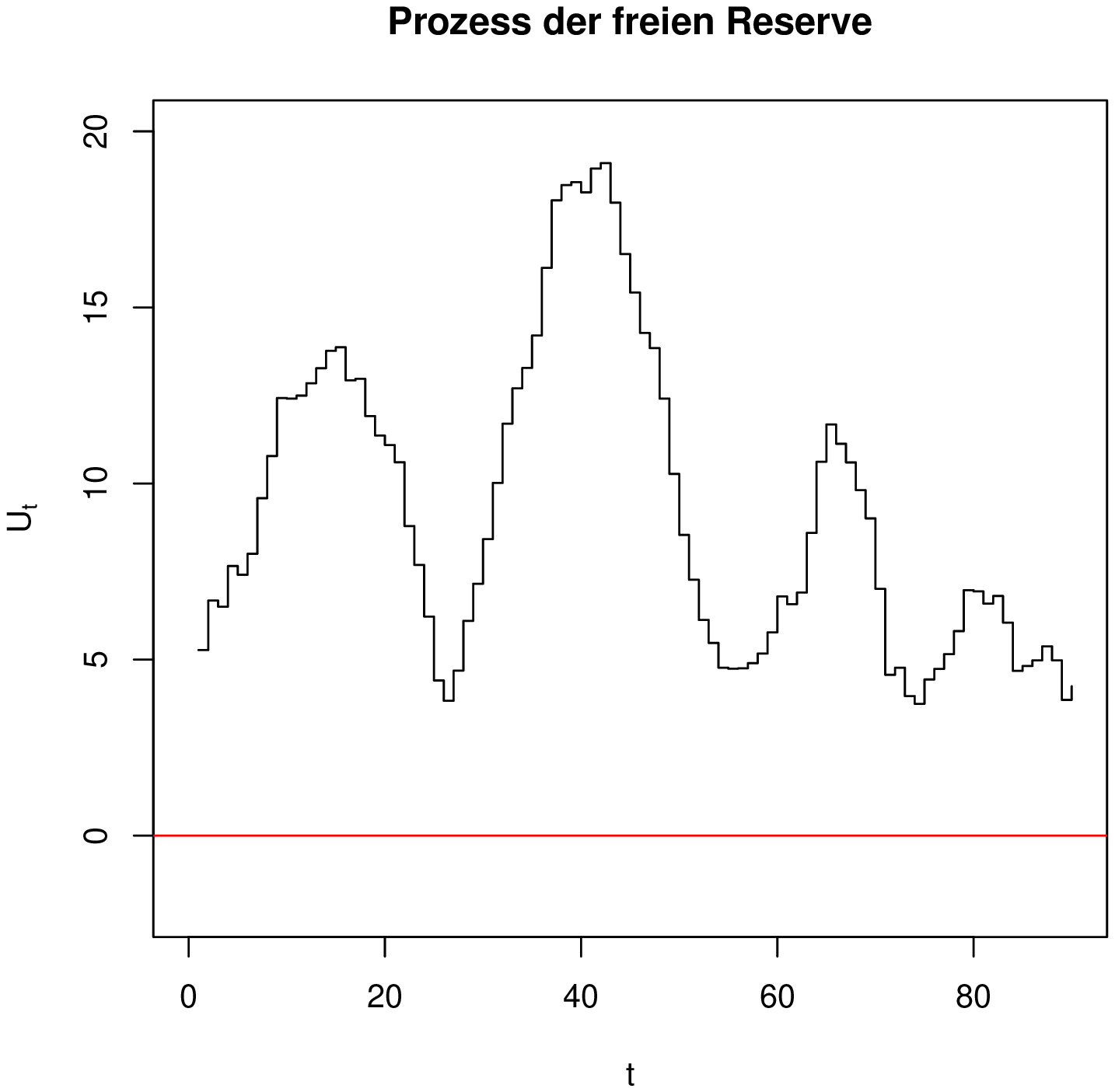}}
\subfigure[Predicted probability of ruin]{\includegraphics[width=0.45\textwidth, trim = 0 0 0 38,clip]{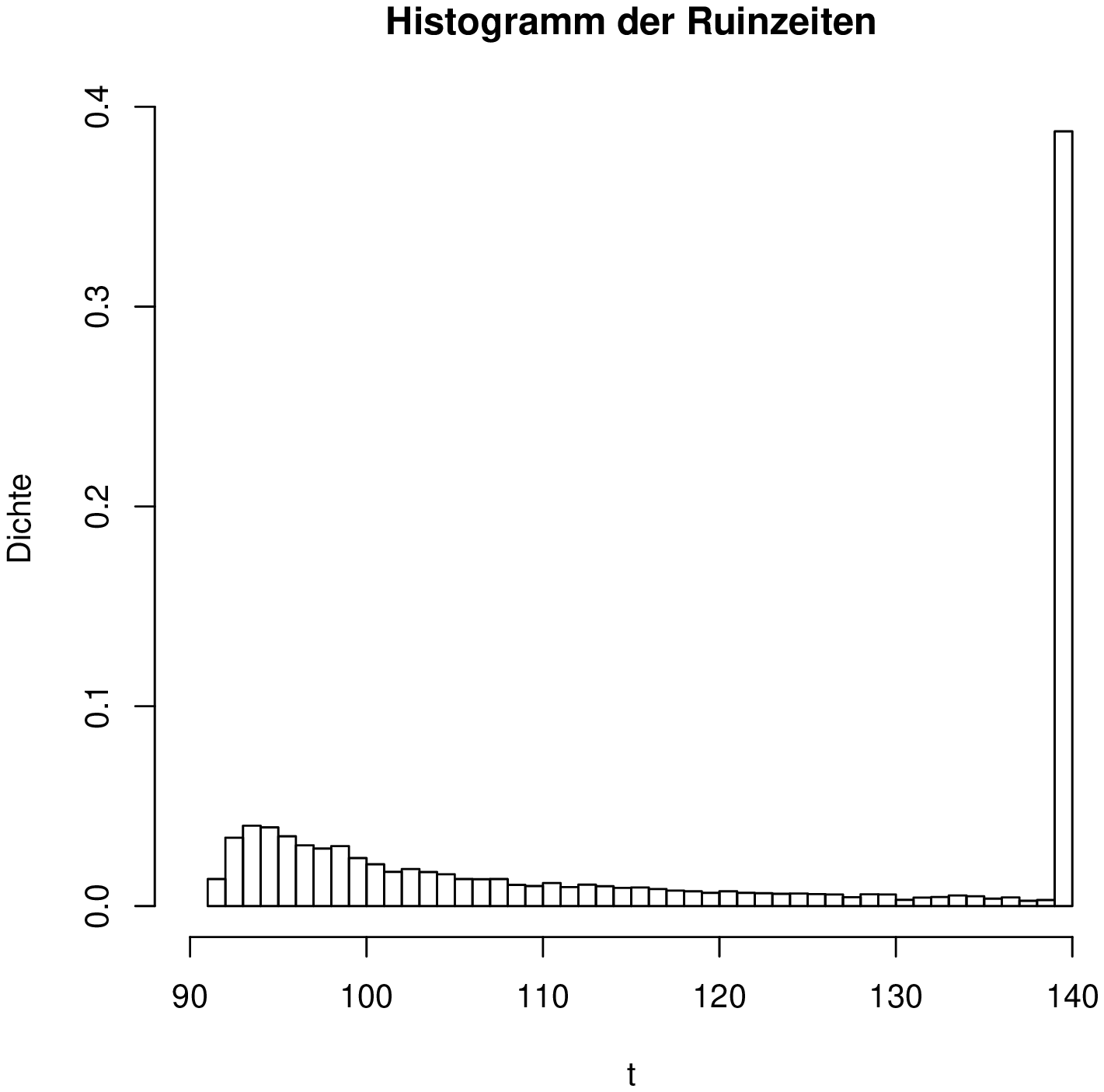}}
\caption{Fictitious surplus process $U_t$ of the last 90 weeks of an insurance company and its predicted probability of ruin}
\label{fig:anw3}
\end{center}
\end{figure}

The prediction of the probability of ruin is done the same way like the forecast of stock markets. For simulating the data the values $u=4$ and $r=0.3$ were used and there are two jumps in the covariance matrix with $c_1=0.9$ at Lag 1 and $c_2=0.8$ at Lag 30. Again, the values are transformed into standard normal distributed random variables and completed by standard normal random numbers. To be able to calculate the probability of ruin for the next 50 weeks, 50 random numbers have to be added. After that, the data are back transformed to correlated variables. Summing up the vector cumulatively and adding $u$, leads to a predicted path. If this procedure is repeated 10,000 times, the points at which $U_t$ gets negative can be illustrated with a histogram. Such a histogram is shown in Figure \ref{fig:anw3}(b) and illustrates the predicted probability of ruin $\psi(u,t)$.
The value at the end of the prediction interval expresses the probability that the insurance company will not have a ruin in the next 50 weeks. To be able to compare the results, the data from Figure \ref{fig:anw3} are analyzed  as they if were uncorrelated. This corresponds to the application of the Cram\'{e}r-Lundberg-Model. The histogram of the simulated points of ruin is shown in Figure \ref{fig:anw5} and illustrates that the probability of ruin is underestimated at the beginning compared to Figure \ref{fig:anw3}(b).

\begin{figure}[!ht]
\begin{center}
\subfigure[Predicted probability of ruin with uncorrelated data]{\includegraphics[width=0.45\textwidth, trim = 0 0 0 38,clip]{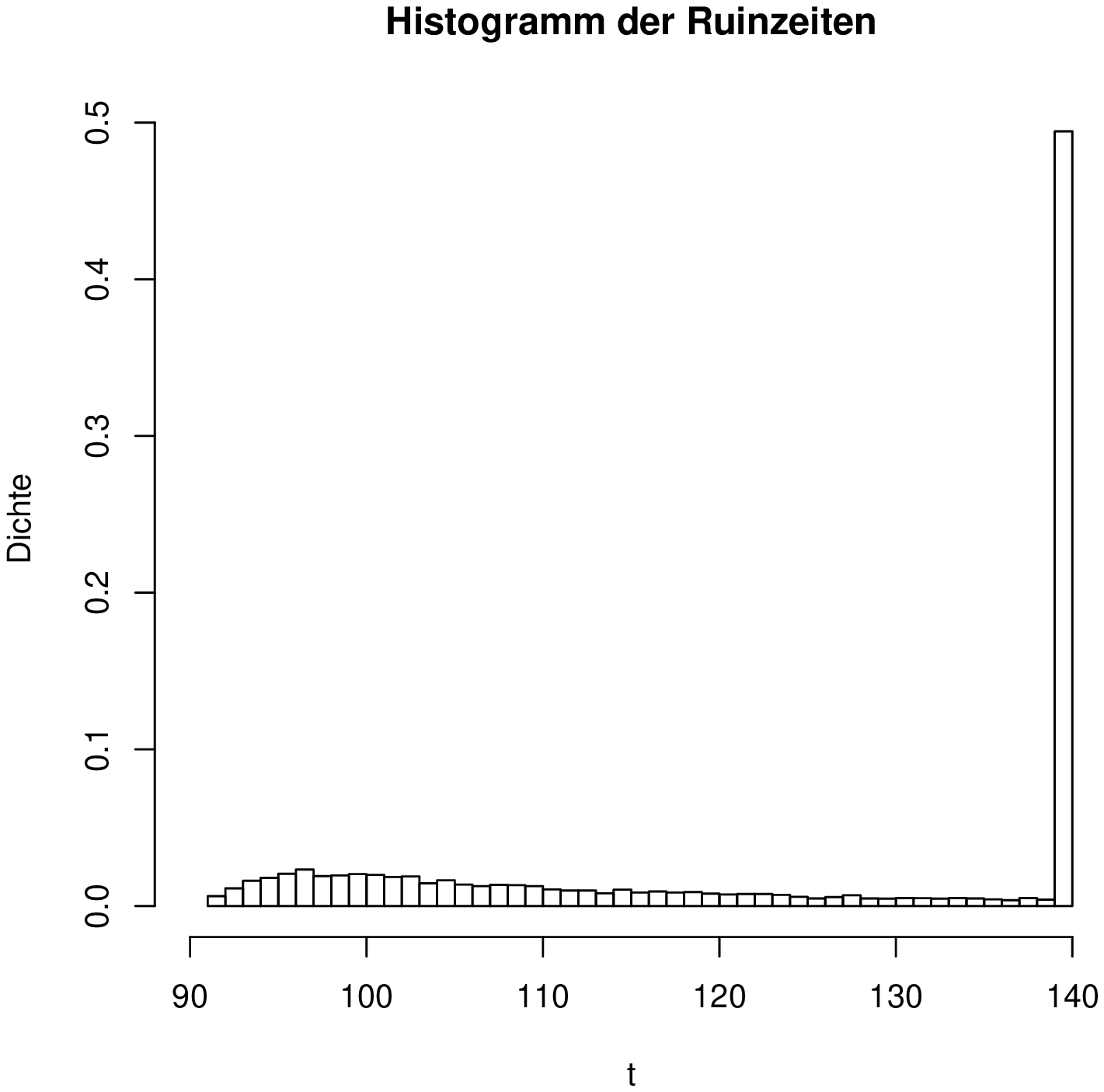}}
\subfigure[Quotients of the predicted probabilities of ruin]{\includegraphics[width=0.45\textwidth, trim = 0 0 0 38,clip]{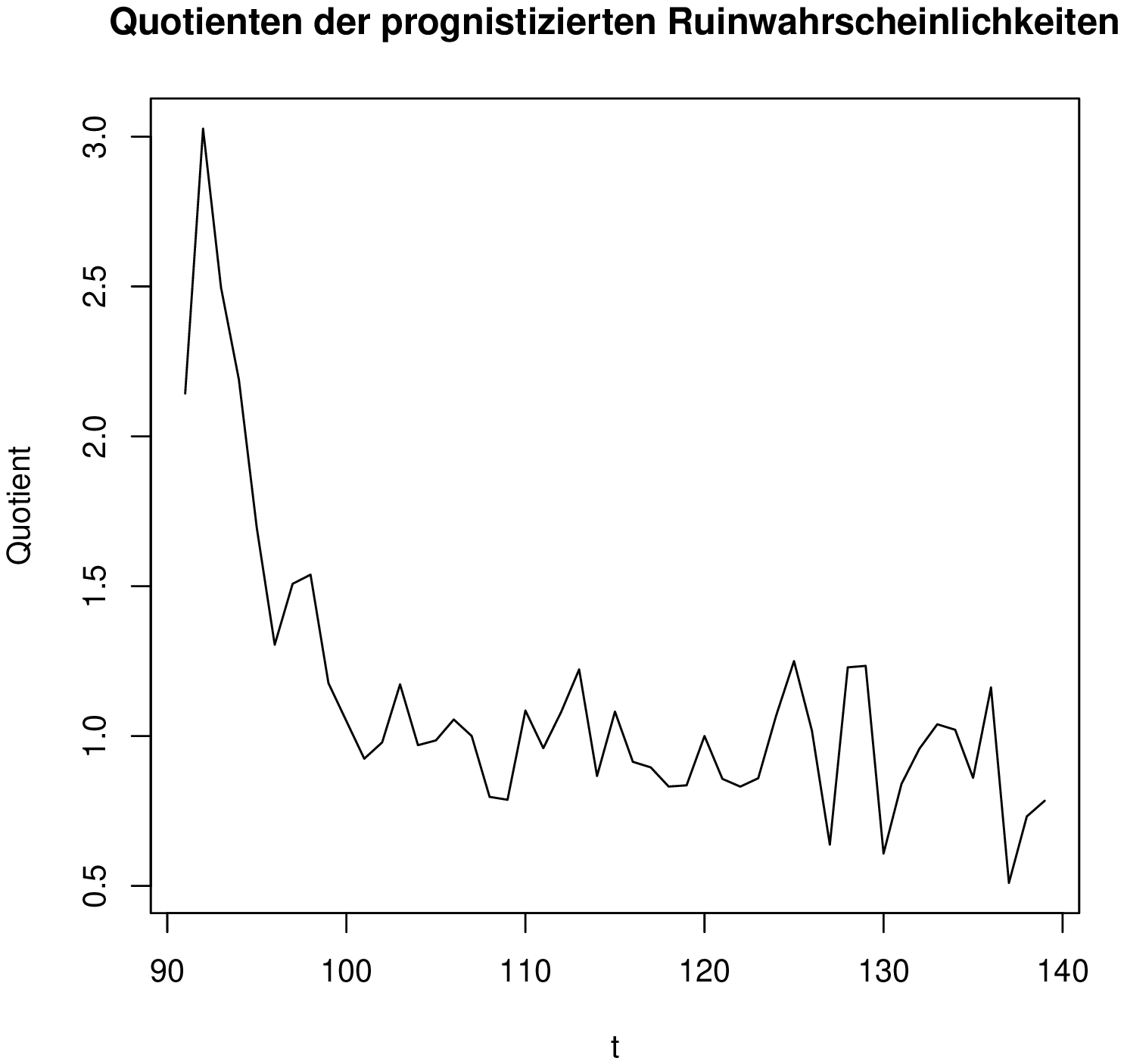}}
\caption{Predicted probability and quotients of predicted probabilites of ruin with uncorrelated data}
\label{fig:anw5}
\end{center}
\end{figure}

The quotients $Q$ of the predicted probabilities of ruin are shown in Figure \ref{fig:anw5}(b) and calculated via
\begin{equation} 	Q=\frac{\psi(u,t)_k}{\psi(u,t)_{ua}} \end{equation}

\section{Conclusions}

Recently, semicontinuous covariance functions have been used by many authors. However, an appropriate discussion on the regularity conditions and statistical properties is up to the best knowledge of the authors still missing. As we have shown in the paper covariance function satisfying conditions abc still possesses some important features of the continuous covariance, justifying the increasing domain asymptotics, e.g.~equidistant optimal design for trend parameter. Moreover, it may have advantage of new desirable features, e.g.~possible existence of admissible designs for correlation parameter. Also in class abc compromise designs seems to be appropriate, like GPD. Another possibility is to construct compound designs. As has been discussed in \cite{Meyer}, kriging is providing discontinuous surfaces even for smooth covariance functions. Therefore a natural idea may appear to employ semicontinuous covariance functions from class abc, which may be more flexible for such modeling.

Due to the fact that the hypothesis of efficient markets of \cite{Fama} is doubted by results of \cite{Lim2010}, a new random walk model for finance data was developed in this paper. For this reason a stationary Ornstein-Uhlenbeck process was used and its covariance matrix modified. This modification enables to adapt the process as good as possible to real data and moreover discontinuous covariance structure.  Parameter $r$ affects the strength of dependency and $c$ permits jumps in covariance structure.  It was shown that both parameters have big impact on resulting paths of the processes. Generally valid formulations were not possible to be done, because the necessary property of positive semi-finiteness of the covariance matrix was not given for all combinations of $r$ and $c$.

In practical applications it is hard to prove the jumps in covariance structure. This could be difficult especially in small samples, however it does not mean that no jumps are existing in the covariance. This is why a test-statistic was suggested to check for a discontinuous covariance matrix. Distribution of the sum of the residuals with one jump, known jump discontinuity and fixed $r$ only enables assumptions about the continuity if distribution function of jump height $c$ is known. More jumps lead to higher complexity and a solution of this problem needs further research in this area.

\section{Acknowledgements}    This work has been partially supported by project ANR project Desire FWF I 833-N18 and Fondecyt Proyecto Regular No 1151441.

\bibliographystyle{elsarticle-harv} 
\bibliography{Semicontinuity}

\begin{thebibliography}{37}
\expandafter\ifx\csname natexlab\endcsname\relax\def\natexlab#1{#1}\fi
\expandafter\ifx\csname url\endcsname\relax
  \def\url#1{\texttt{#1}}\fi
\expandafter\ifx\csname urlprefix\endcsname\relax\def\urlprefix{URL }\fi

\bibitem[{Abt and Welch(1998)}]{Abt}
Abt, M., Welch, W.~J., 1998. {Fisher information and maximum-likelihood
  estimation of covariance parameters in Gaussian stochastic processes}.
  Canadian Journal of Statistics-Revue Canadienne De Statistique 26~(1),
  127--137.

\bibitem[{Baran et~al.(2013)Baran, Sikolya, and Stehl{\'{\i}}k}]{BaranSS}
Baran, S., Sikolya, K., Stehl{\'{\i}}k, M., 2013. {On the optimal designs for
  predic- tion of Ornstein-Uhlenbeck sheets}. Statistics {\&} Probability
  Letters 83~(6), 1580--1587.

\bibitem[{Belyaev(1961)}]{BELAYEV}
Belyaev, Y.~K., 1961. {Continuity and H{\"{o}}lder’s conditions for sample
  functions of stationary Gaussian processes}. In: Proc. 4th Berkeley Symp.
  Math. Statist. and Prob. Vol. 2. pp. 22--33.

\bibitem[{Crary and Stormann(2015)}]{Crary15}
Crary, S., Stormann, J., 2015. {Four-Point, 2D, Free-Ranging, IMSPE-Optimal,
  Twin-Point Designs}.

\bibitem[{Crary(2002)}]{Crary02}
Crary, S.~B., 2002. {Design of Computer Experiments for Metamodel Generation}.
  Analog Integrated Circuits and Signal Processing 32~(1), 7--16.

\bibitem[{Cressie(1993)}]{Cressie}
Cressie, N., 1993. {Statistics for Spatial Data (revised edition)}.

\bibitem[{Crum(1956)}]{Crum}
Crum, M.~M., 1956. {On positive-definite functions}. In: Proc. London Math.
  Soc.(3). Vol. 6. pp. 548--560.

\bibitem[{Fama(1965)}]{Fama}
Fama, E.~F., 1965. {Random Walks in Stock Market Prices}. Financial Analysts
  Journal 21~(5), 55--59.

\bibitem[{Fama(1970)}]{Fama1970}
Fama, E.~F., 1970. {Efficient Capital Markets - A Review of Theory and
  Empirical Work.pdf}. The Journal of Finance 25~(2), 36.

\bibitem[{Girdziu{\v{s}}as and Laaksonen(2005)}]{Girdziusas}
Girdziu{\v{s}}as, R., Laaksonen, J., 2005. {Use of input deformations with
  brownian motion filters for discontinuous regression}. In: Pattern
  Recognition and Data Mining. Springer, pp. 219--228.

\bibitem[{Hoel(1958)}]{Hoel}
Hoel, P.~G., 1958. {Efficiency Problems in Polynomial Estimation}. The Annals
  of Mathematical Statistics 29~(4), 1134--1145.

\bibitem[{Horn and Johnson(1994)}]{Horn2}
Horn, R., Johnson, C., 1994. {Topics in Matrix Analysis}. University Press,
  Cambridge.

\bibitem[{Horn and Johnson(2006)}]{Horn}
Horn, R., Johnson, C., 2006. {Matrix Analysis}. University Press, Cambridge.

\bibitem[{Kisel{\'{a}}k and Stehl{\'{\i}}k(2008)}]{KS}
Kisel{\'{a}}k, J., Stehl{\'{\i}}k, M., 2008. {Equidistant and D-optimal designs
  for parameters of Ornstein-Uhlenbeck process?} Statistics {\&} Probability
  Letters 78~(12), 1388--1396.

\bibitem[{Kupka(2000)}]{Kupka00}
Kupka, I., 2000. {Convergences preserving the fixed point property}.
  Mathematica Slovaca 50~(4), 483--494.

\bibitem[{Kupka and Toma(1995)}]{KupkaToma}
Kupka, I., Toma, V., 1995. {A uniform convergence for non-uniform spaces}.
  Publicationes mathematicae-Debrecen 47~(3-4), 299--309.

\bibitem[{Lim and Brooks(2010)}]{Lim2010}
Lim, K.-P., Brooks, R.~D., 2010. {Why do emerging stock markets experience more
  persistent price deviations from a random walk over time? A country-level
  analysis}. Macroeconomic Dynamics 14~(S1), 3--41.

\bibitem[{Matheron(1963)}]{Matheron}
Matheron, G., 1963. {Principles of geostatistics}. Economic geology 58~(8),
  1246--1266.

\bibitem[{Meyer(2004)}]{Meyer}
Meyer, T.~H., 2004. {The Discontinuous Nature of Kriging Interpolation for
  Digital Terrain Modeling}. Cartography and Geographic Information Science
  31~(July 2014), 209--216.

\bibitem[{M{\"{u}}ller and P{\'{a}}zman(1999)}]{MuellerPazman}
M{\"{u}}ller, W.~G., P{\'{a}}zman, A., 1999. {An algorithm for the computation
  of optimum designs under a given covariance structure}. Computational
  Statistics 14~(2), 197--211.

\bibitem[{M{\"{u}}ller and Stehl{\'{\i}}k(2010)}]{MullerStehlik3}
M{\"{u}}ller, W.~G., Stehl{\'{\i}}k, M., 2010. {Compound optimal spatial
  designs}. Environmetrics 21~(April 2009), 354--364.

\bibitem[{N{\"{a}}ther(1985)}]{Nather85}
N{\"{a}}ther, W., 1985. {Effective observation of random fields}. Vol.~72.
  Teubner.

\bibitem[{Neubrunn(1988)}]{Neubrunn}
Neubrunn, T., 1988. {Quasi-continuity}. Real Anal. Exchange 14~(2), 259--306.

\bibitem[{Neubrunnov{\'{a}}(1974)}]{Neubrunnova}
Neubrunnov{\'{a}}, A., 1974. {On quasicontinuous and cliquish functions}.
  {\v{C}}asopis pro p{\v{e}}stov{\'{a}}n{\'{\i}} matematiky 99~(2), 109--114.

\bibitem[{P{\'{a}}zman(2007)}]{Pazman07}
P{\'{a}}zman, A., 2007. {Criteria for optimal design of small-sample
  experiments with correlated observations}. Kybernetika 43~(4), 453--462.

\bibitem[{P{\'{o}}lya and Others(1949)}]{Polya}
P{\'{o}}lya, G., Others, 1949. {Remarks on characteristic functions}. In: Proc.
  First Berkeley Conf. on Math. Stat. and Prob. pp. 115--123.

\bibitem[{Pukelsheim(1993)}]{Pukelsheim}
Pukelsheim, F., 1993. {Optimal design of experiments}. Vol.~50. siam.

\bibitem[{{R Development Core Team}(2014)}]{R}
{R Development Core Team}, 2014. {R: A Language and Environment for Statistical
  Computing}.

\bibitem[{Rao(2009)}]{Rao}
Rao, C.~R., 2009. {Linear statistical inference and its applications}. Vol.~22.
  John Wiley {\&} Sons.

\bibitem[{Sacks et~al.(1989)Sacks, Welch, Mitchell, and Wynn}]{SacksW}
Sacks, J., Welch, W.~J., Mitchell, T.~J., Wynn, H.~P., 1989. {Design and
  Analysis of Computer Experiments}. Statistical Science 4~(4), 409--423.

\bibitem[{Sacks and Ylvisaker(1966)}]{SacksY66}
Sacks, J., Ylvisaker, D., 1966. {Designs for regression problems with
  correlated errors}. The Annals of Mathematical Statistics, 66--89.

\bibitem[{Stehlik(2004)}]{Stehlik04}
Stehlik, M., 2004. {Further aspects on an example of D-optimal designs in the
  case of correlated errors}.

\bibitem[{Stehl{\'{\i}}k(2014)}]{Stehlik2014}
Stehl{\'{\i}}k, M., 2014. {On convergence of topological aggregation
  functions}. Fuzzy Sets and Systems.

\bibitem[{Stehl{\'{\i}}k et~al.(2008)Stehl{\'{\i}}k,
  Rodr{\'{\i}}guez-D{\'{\i}}az, M{\"{u}}ller, and L{\'{o}}pez-Fidalgo}]{Test}
Stehl{\'{\i}}k, M., Rodr{\'{\i}}guez-D{\'{\i}}az, J.~M., M{\"{u}}ller, W.~G.,
  L{\'{o}}pez-Fidalgo, J., 2008. {Optimal allocation of bioassays in the case
  of parametrized covariance functions: an application to lung’s retention of
  radioactive particles}. Test 17~(1), 56--68.

\bibitem[{Yuen and Guo(2001)}]{Yuen2001}
Yuen, K.~C., Guo, J.~Y., 2001. {Ruin probabilities for time-correlated claims
  in the compound binomial model}. Insurance: Mathematics and Economics 29~(1),
  47--57.

\bibitem[{Zagoraiou and {Baldi Antognini}(2009)}]{Zagoraiou}
Zagoraiou, M., {Baldi Antognini}, A., 2009. {Optimal designs for parameter
  estimation of the Ornstein--Uhlenbeck process}. Applied Stochastic Models in
  Business and Industry 25~(5), 583--600.

\bibitem[{Zhu and Stein(2005)}]{Zhu}
Zhu, Z., Stein, M.~L., 2005. {Spatial sampling design for parameter estimation
  of the covariance function}. Journal of Statistical Planning and Inference
  134, 583--603.

\end{thebibliography}

\appendix
\section{Proofs}

\begin{proof}{Lemma 1}

Let us have $n=2,$ i.e. we have two point design $x_1,x_2.$ Then optimal design for parameters $(\theta, r)$ is collapsing (see \cite{KS}), i.e.~$M_rM_\theta$ attains its maximum at $x_1=x_2.$

Let us have $n>2.$ Then we have $\frac{\partial F}{\partial d_i}<0.$ Therefore also the directional derivative is negative in all canonical directions with the start at the beginning of coordinate system. Therefore  $M_rM_\theta$ attains its maximum at $d_1=d_2=...=d_{n-1},\sum d_i=1,$ so the optimal design is equidistant.
\end{proof}

\begin{proof}{Theorem 1}

 Pointwise convergence of maps defined  on a locally separable metric space mapping to a metric space preserves the semicontinuity of maps (see \cite{Neubrunnova}). Thus $K\geq 0$ is non-increasing, and continuous. The interchange of the limits justifies c), i.e. $\lim_{d\to +\infty}C( d)=\lim_{d\to +\infty}\lim_{n\to +\infty}K_n(d,0) = $ \\ $\lim_{n\to +\infty} \lim_{d\to +\infty} K_n(d,0) = 0.$
\end{proof}

\begin{proof}{Lemma 2}

Condition a) is satisfied, since strong convergence implies the pointwise convergence when target space is regular in topological space (see \cite{KupkaToma}) and pointwise convergence preserves inequalities.

Condition b) is satisfied since strong convergence preserves continuity when target space is a regular topological space (see \cite{KupkaToma}).

Condition c) is satisfied, since the regularity of target space guaranties that strong convergence preserves continuity and implies pointwise convergence. Using these facts we finally get \\
 $\lim_{d\to +\infty}C( d)=\lim_{d\to +\infty}\lim_{\gamma\in \Gamma}K_\gamma(d,0)=\lim_{\gamma\in \Gamma}\lim_{d\to +\infty}K_\gamma(d,0)=0.$
\end{proof}

\begin{proof}{Theorem 3}

Condition a) follows from Theorem 5.

Condition b)  is satisfied because of Lemma 2 and regularity of target space.

Condition c) is satisfied because of Lemma 3, regularity and fully normality of target
space.
\end{proof}

\begin{proof}{Theorem 4}

 First, let us recall the Frobenius theorem (see \cite{Rao}, p.46). An irreducible positive matrix $A$ always has a positive characteristic value $\lambda_0(A)$ which is a simple root of the characteristic equation and not smaller than the moduli of other characteristic values. Moreover, if $A\geq B\geq 0$ then $\lambda_0(A)\geq \lambda_0(B).$

Now let $+\infty>d_1>d_2\geq 0.$ Then $C_{i,j}\left( d_1,r \right)\leq C_{i,j}\left( d_2,r \right)$ for all $i,j=1,..,n$ and thus $C\left( d_2,r \right)\geq C\left( d_1,r \right)\geq 0.$ Employing the Frobenius theorem we have $\lambda_0(C\left( d_2,r \right))\geq \lambda_0(C\left( d_1,r \right)).$ Our matrix is symmetric and real, thus we have

$\lambda_{min}(C^{-1}\left( d_2,r \right))\leq \lambda_{min}(C^{-1}\left( d_1,r \right)),$ where $\lambda_{min}(A)$ denotes the minimal eigenvalue of matrix $A.$

Now, $M_\theta(n)=1^TC^{-1}\left( d,r \right)1\geq n\inf_{x}\frac{x^TC^{-1}\left( d,r \right)x}{x^Tx}=\lambda_{min}(C^{-1}\left( d,r \right))$ and thus we have proven that for an equidistant design the lower bound function $d\to n\inf_{x}\frac{x^TC^{-1}\left( d,r \right)x}{x^Tx}$
is non decreasing. Similarly we can prove the rest of 1). Now let $C(d)$ be decreasing. We will show that $M_\theta$ is increasing.
Notice that $\lambda_{max}(C^{-1})=\rho(C^{-1})$ is decreasing with $d,$ more precisely if $0\leq A\leq B$ then $\rho(A)\leq \rho(B)$
and if $0\leq A< B$ and $A+B$ is irreducible, then $\rho(A)< \rho(B)$ (see \cite{Horn2}). Let us have $+\infty>d_1>d_2\geq 0.$
Then $C(d_1)<C(d)$ and we have $\rho(C^{-1}(d_2))<\rho(C^{-1}(d_1)).$  Let $\epsilon=\frac{\rho(C^{-1}(d_1))-\rho(C^{-1}(d_2))}2,$  then
 we have (see Lemma 5.6.10 in \cite{Horn}) such a matrix norm  $||.||_\star$ that $\rho(C^{-1}(d_2))\leq
 ||C^{-1}(d_2)||_\star<\rho(C^{-1}(d_2))+\epsilon<\rho(C^{-1}(d_1))\leq  ||C^{-1}(d_1)||_\star.$
Here we use the fact that $\rho(A)=\inf\{||A||,||.|| {\rm is\ a\ matrix\ norm}\}.$  Thus we have $||C^{-1}(d_2)||_\star<
||C^{-1}(d_1)||_\star$  and norms $||.||_\star$ and l-1 norm $||.||_1=\sum_{i,j} |A_{i,j}|$ are equivalent and
 $M_\theta(d)=||C^{-1}(d)||_1$. Thus we have $M_\theta(d_2)<M_\theta(d_1).$

To prove ii) let us consider the open set $U$ of all covariance matrices $C_r(d)$ with bounded inverse in a Banach space of real matrices $n\times n.$ Then the identity $I(n)=\lim_{\forall i:d_i\to +\infty}C_r(d)\in U$ and map $C(n) \to C(n)^{-1}$ is smooth. This implies
$$a(n,n-1)(+\infty)=\lim_{\forall i:d_i\to +\infty} \frac{1^TC(n)^{-1}_r(d)1}{1^TC(n-1)^{-1}_r(d)1}=\frac{1^TI(n)1}{1^TI(n-1)1}=\frac{n}{n-1}.\square$$

To prove iii) let us  consider representation $C(d)=\exp(-\psi_r(d)).$ Process considered in distances $\xi_i=\psi_r(d_i)$ is Ornstein Uhlenbeck and for such a process equidistant design is optimal (see \cite{KS}) and all neighboring point distances $\xi_i$ increase with same speed. $\psi_r(d)$ is nondecreasing function, therefore all neighboring point distances  $d_i$ of original process  increase with same speed. Therefore also for the original process (with covariance $C(d)$) is equidistant design optimal.

iv) Optimal design for parameter $r$ in stationary Ornstein Uhlenbeck process is collapsing (see \cite{KS} and \cite{Zagoraiou}). Process considered in distances $\xi_i=\psi_r(d_i)$ is Ornstein Uhlenbeck.  $\psi_r(d)$ is nondecreasing function with possible jump at point 0. Therefore it may exist an  admissible optimal design (nugget effect).
\end{proof}

\end{document}